\documentclass[conference]{IEEEtran}

\usepackage{cite}
\usepackage{amsmath,amssymb,amsfonts}
\usepackage{algorithmic}
\usepackage{graphicx}
\usepackage{textcomp}
\usepackage{xcolor}

\usepackage{booktabs}
\usepackage{algorithm}
\usepackage{algorithmic}
\usepackage{amssymb}
\usepackage{amsmath}
\usepackage{amsthm}
\usepackage{subfigure}

\begin{document}

\title{Learning Topological Representation for Networks via Hierarchical Sampling \\
\thanks{Identify applicable funding agency here. If none, delete this.}
}

\author{\IEEEauthorblockN{Guoji~Fu, Chengbin~Hou, and Xin~Yao} \\

	\IEEEauthorblockA{Shenzhen Key Lab of Computational Intelligence \\
		University Key Laboratory of Evolving Intelligent Systems of Guangdong Province \\
		Department of Computer Science and Engineering \\
		Southern University of Science and Technology, Shenzhen, 518055, P. R. China}
	Email: fugj2017@student.sustc.edu.cn, chengbin.hou10@foxmail.com, xiny@sustc.edu.cn}
\maketitle

\begin{abstract}
The topological information is essential for studying the relationship between nodes in a network. Recently, \textbf{N}etwork \textbf{R}epresentation \textbf{L}earning (NRL), which projects a network into a low-dimensional vector space, has been shown their advantages in analyzing large-scale networks. However, most existing NRL methods are designed to preserve the local topology of a network, they fail to capture the global topology. To tackle this issue, we propose a new NRL framework, named HSRL, to help existing NRL methods capture both the local and global topological information of a network. Specifically, HSRL recursively compresses an input network into a series of smaller networks using a community-awareness compressing strategy. Then, an existing NRL method is used to learn node embeddings for each compressed network. Finally, the node embeddings of the input network are obtained by concatenating the node embeddings from all compressed networks. Empirical studies for link prediction on five real-world datasets demonstrate the advantages of HSRL over state-of-the-art methods.
\end{abstract}


\section{Introduction}
The science of networks has been widely used to understand the behaviours of complex systems. These systems are typically described as networks, such as social networks in social media \cite{scott2017social}, bibliographic networks in academic field \cite{sun2011pathsim}, protein-protein interaction networks in biology \cite{theocharidis2009network}. Studying the relationship between entities in a complex system is an essential topic, which benefits a variety of applications \cite{goyal2018graph}. Just take a few examples, predicting potential new friendship between users in social networks \cite{yang2011like}, searching similar authors in bibliographic networks \cite{sun2011pathsim}, recommending new movies to users in movie user-movie interest networks \cite{diao2014jointly}. The topologies of these networks provide insight information on the relationship between nodes. We can find out strongly connected neighborhoods of a node by exploring the local topology in a network. Meanwhile, the global topology is another significant aspect for studying the relationship between communities. As shown in Fig.\ref{figure1}, such hierarchical topological information is helpful to learn the relationship between nodes in a network.

\begin{figure}
	\centering
	\includegraphics[scale=0.44]{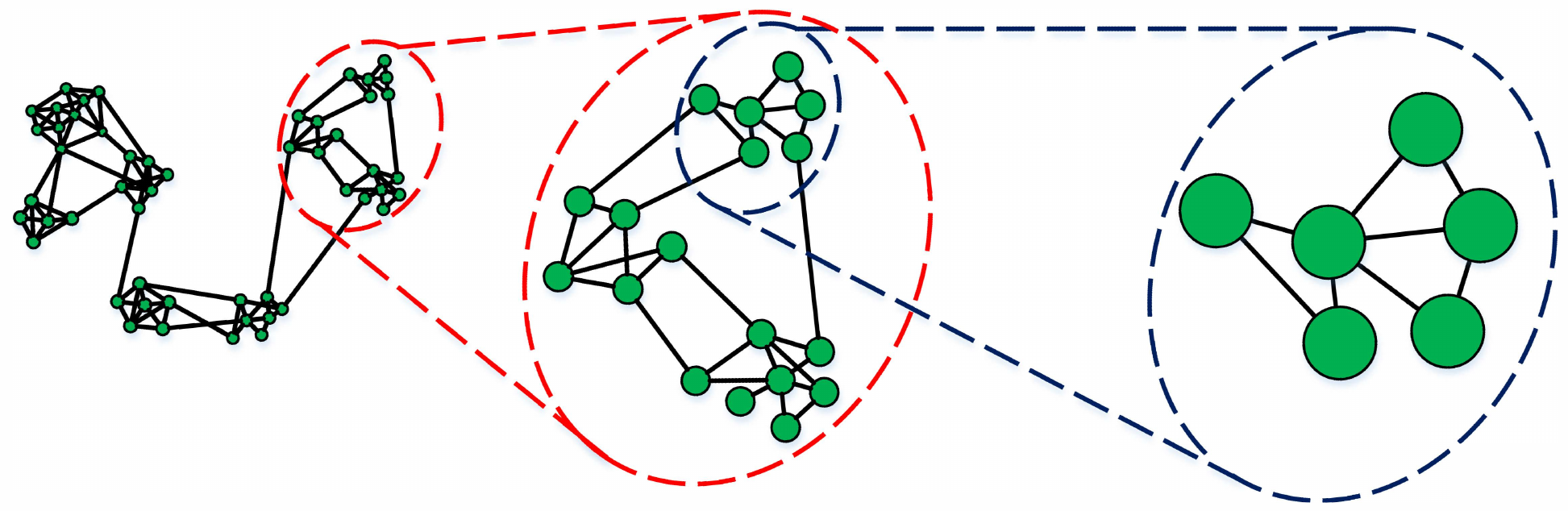}
	\caption{An example of hierarchical view of network topology.}
	\label{figure1}
\end{figure}

Many networks are large-scale in real-world scenarios, such as a Facebook social network contains billion of users \cite{ellison2007benefits}. As a result, most traditional network analytic methods suffer from high computation and space cost \cite{goyal2018graph}. To tackle this issue, \textbf{N}etwork \textbf{R}epresentation \textbf{L}earning (NRL) has been a popular technique to analyze large-scale networks recently. In particular, NRL aims to map a network into a low-dimensional vector space, while preserving as much of the original network topological information as possible. Nodes in a network are represented as low-dimensional vectors which are used as input features for downstream network analysis algorithms.

Traditional NRL methods such as LLE \cite{roweis2000nonlinear} and ISOMap \cite{tenenbaum2000global} work well on small networks, while they are infeasible to large-scale networks due to the high computational cost. Recently, some online learning methods, e.g., DeepWalk \cite{perozzi2014deepwalk}, node2vec \cite{grover2016node2vec}, and LINE \cite{tang2015line}, have been proposed to learn large-scale network representation, which has been demonstrated their efficiency and effectiveness for the large-scale network analysis.

However, the above NRL methods only consider the local topology of networks and fail to capture the global topological information. DeepWalk and node2vec firstly employ short random walks to explore the local neighborhoods of nodes and obtain node embeddings by the Skip-Gram model \cite{mikolov2013distributed}. LINE preserves the first-order and second-order proximities so that it can only measure the relationship between nodes at most two-hops away. These methods are efficient to capture the relationship between close nodes, however, fail to consider the case for nodes which are far away from each other. Recently, HARP \cite{chen2018harp} has been proposed to overcome this issue. It recursively compresses a network into a series of small networks based on two node collapsing schemes and learns node embeddings for each compressed network by using an existing NRL method. Unfortunately, the compressed networks may not reveal the global topology of an input network, since HARP heuristically merges two closed nodes into a new node. Furthermore, when learning node embeddings on the original network, using node embeddings obtained on compressed networks as the initialization solution may mislead the optimization process to a bad local minimum.

This paper presents a new NRL framework, called \textbf{H}ierarchical \textbf{S}ampling \textbf{R}epresentation \textbf{L}earning (HSRL), to learn node embeddings for a network with preserving both their local and global topological information. Specifically, HSRL uses a community-awareness network compressing strategy, called hierarchical sampling, to recursively compress an input network into a series of smaller networks, and then engage an existing NRL method to learn node embeddings for each compressed network. Finally, the node embeddings of the original network can be obtained by concatenating all node embeddings learned on compressed networks. Besides, we mathematically show that HSRL is able to capture the local and global topological relationship between nodes. Novel contributions of this paper include the following:
\begin{itemize}
	\item We propose a new NRL framework called HSRL, to learn node embeddings for a network, which is able to capture both local and global topological information of the network via a community-awareness network compressing strategy.
	
	\item We mathematically show that the node embeddings obtained by HSRL explicitly embed the local and global topological information of the input network.
	
	\item We demonstrate that HSRL statistically significantly outperforms DeepWalk, node2vec, LINE, and HARP on link prediction tasks on five real-world datasets.
\end{itemize}

\section{Related Work}
Most early methods in NRL field represent an input network in the form of a matrix, e.g., adjacency matrices \cite{roweis2000nonlinear, ahmed2013distributed}, Laplacian matrices \cite{belkin2002laplacian}, node transition probability matrices \cite{cao2015grarep}, and then factorize that matrix to obtain node embeddings. They are effective for small networks, but cannot scale to large-scale networks due to high computation cost.

To analyze large-scale networks, DeepWalk \cite{perozzi2014deepwalk} employs truncated random walks to obtain node sequences, and then learns node embeddings by feeding node sequences into Skip-Gram model \cite{ahmed2013distributed}. To generalize DeepWalk, node2vec \cite{grover2016node2vec} provides a trade-off between breadth-first search (BFS) and depth-first search (DFS) when generating truncated random walks for a network. LINE \cite{tang2015line} intends to preserve first-order and second-order proximities, respectively, by minimizing the Kullback-Leibler divergence of two joint probability distributions for each pair nodes. These methods are scalable to large-scale networks, but fail to capture the global topological information of networks. Because random walks are only effective to explore local neighborhoods for a node, and both first-order and second-order proximities defined by LINE just measure the relationship between nodes at most two-hops away.

To investigate global topologies of a network, HARP \cite{chen2018harp} recursively uses two collapsing schemes, edge collapsing and star collapsing, to compress an input network into a series of small networks. Starting from the smallest compressed network, it then recursively conducts a NRL method to learn node embeddings based on the node embeddings obtained from its previous level (if any) as the initialization. However, HARP has two weaknesses: 1) nodes that are connected but belong to different communities may be merged, which leads to that the compressed networks cannot well reveal the global topology of an input network. 2) taking the node embeddings learned on such compressed networks as initialization would mislead NRL methods to a bad local minimum. HARP could work well on node classification tasks since close nodes tend to have the same labels but may ineffective for the link prediction tasks. Because predicting the link between two nodes needs to consider both the local and global topological information of a network, such as neighborhoods they are sharing with and communities they are both involved in. This paper proposes HSRL to tackle the above issues of the existing NRL methods.

\section{Preliminary and Problem definition}
This section gives the notations and definitions throughout this paper.

We firstly introduce the definition of a network and related notations.
\newtheorem{mydef}{Definition}
\begin{mydef}
	(\textbf{Network}) \cite{goyal2018graph} A network (a.k.a. graph) is defined as $G = (V, E)$, where $V$ is a set of nodes and $E$ is a set of edges between nodes. The edge $e \in E$ between nodes $u$ and $v$ is represented as $e = (u, v)$  with a weight $w_{u,v} \geq 0$. Particularly, we have $(u, v) \equiv (v, u)$ and $w_{u,v} \equiv w_{v,u}$ if $G$ is undirected; $(u, v) \equiv (v, u)$ and $w_{u,v} \not\equiv w_{v,u}$, otherwise.
\end{mydef}

In most networks, some nodes are densely connected to form a community$/$cluster, while nodes in different communities are sparsely connected. Detecting communities in a network is beneficial to analyze the relationship between nodes. We employ modularity as defined below to evaluate the quality of community detection.

\begin{mydef}
    (\textbf{Modularity}) \cite{newman2006modularity, newman2004analysis} Modularity is a measure of the structure of networks, which measures the density of edges between the nodes within communities as compared to the edges between nodes in different communities. It is defined as below.
    \begin{equation}\label{equation1}
        Q = \frac{1}{2m}\sum_{i,j}\big(w_{i,j} - \frac{k_ik_j}{2m}\big)\delta(c_i, c_j)
    \end{equation}
    \noindent where $w_{i,j}$ is the weight of edge $e_{i,j}$ between nodes $v_i$ and $v_j$, $k_i = \sum_iw_{i,j}$, $m=\frac{1}{2}\sum_{i,j}w_{i,j}$, $c_i$ is the community which node $v_i$ belongs to, and \\
    \[ \delta(u, v) =
    \begin{cases}
    1, & \mbox{if u = v},\\
    0, & \mbox{otherwise}.
    \end{cases}
    \]
\end{mydef}

Networks with high modularity have dense connections between nodes within communities but sparse connections between nodes in different communities.

We give the definition of hierarchical sampling which is used to recursively compress a network into a series of smaller networks as follows.

\begin{mydef}
    (\textbf{Hierarchical Sampling}) Given a network $G = (V, E)$, hierarchical sampling compresses the original network level by level and obtains a series of compressed networks $G^0, G^1, ..., G^K$, which reveals the global topological information of original network at different levels, respectively. 
\end{mydef}

These compressed networks reveal the hierarchical topologies of the input network. Therefore, the node embeddings obtained in compressed networks embed the hierarchical topological information of the original network. 

To learn node embeddings of a network, NRL maps the original network into a low-dimensional space and represents each node as a low-dimensional vector as formulated below.
\begin{mydef}
	(\textbf{Network Representation Learning}) \cite{goyal2018graph} Given a network $G = (V, E)$, network representation learning aims to learn a mapping function $f: v \rightarrow z \in \mathbb{R}^d$ where $d \ll |V|$, and preserving as much of the original topological information in the embedding space $\mathbb{R}^d$.
\end{mydef}

Finally, we present the formulation of the hierarchical network representation learning problem as following:
\begin{mydef}
	(\textbf{Hierarchical Network Representation Learning}) Given a series of compressed networks $G^0, G^1, ..., G^K$ of original network $G = (V, E)$ and a network representation learning mapping function $f$, hierarchical network representation learning learns the node embeddings for each compressed network by $Z^k \leftarrow f(G^k), 0\leq k \leq K$, and finally obtains the node embeddings $Z$ of original network $G$ by concatenating $Z^0, Z^1, ..., Z^K$.
\end{mydef}

\section{HSRL}
In this section, we present Hierarchical Sampling Representation Learning framework which consists of two parts: 1) Hierarchical Sampling that aims to discover the hierarchical topological information of a network via a community-awareness compressing strategy; and 2) Representation Learning that aims to learn low-dimensional node embeddings while preserving the hierarchical topological information. 

\begin{figure}
    \centering
    \includegraphics[scale=0.4]{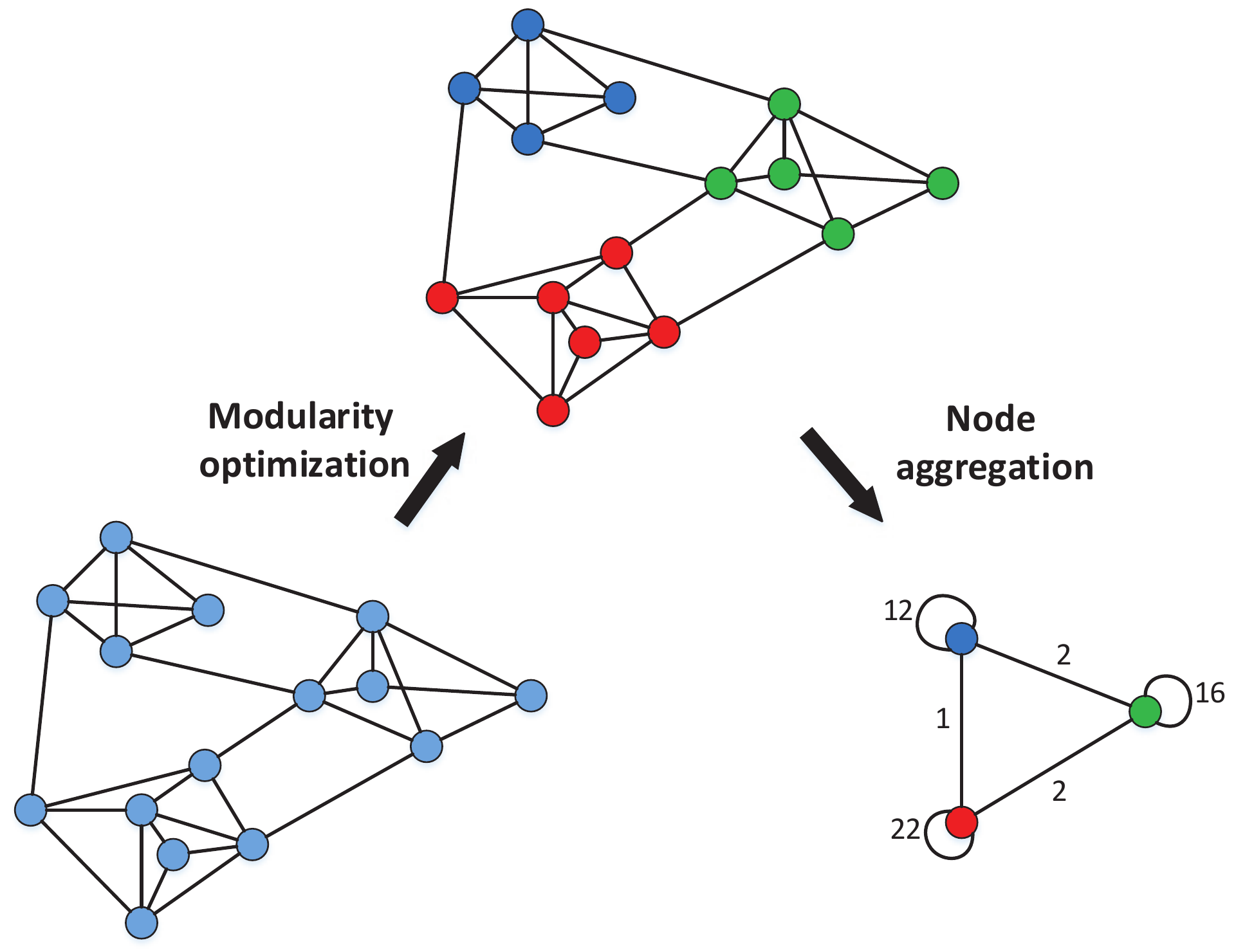}
    \caption{An exampling of compressing a network}
    \label{figure2}
\end{figure}

\begin{figure*}
\centering
	\includegraphics[scale=0.25]{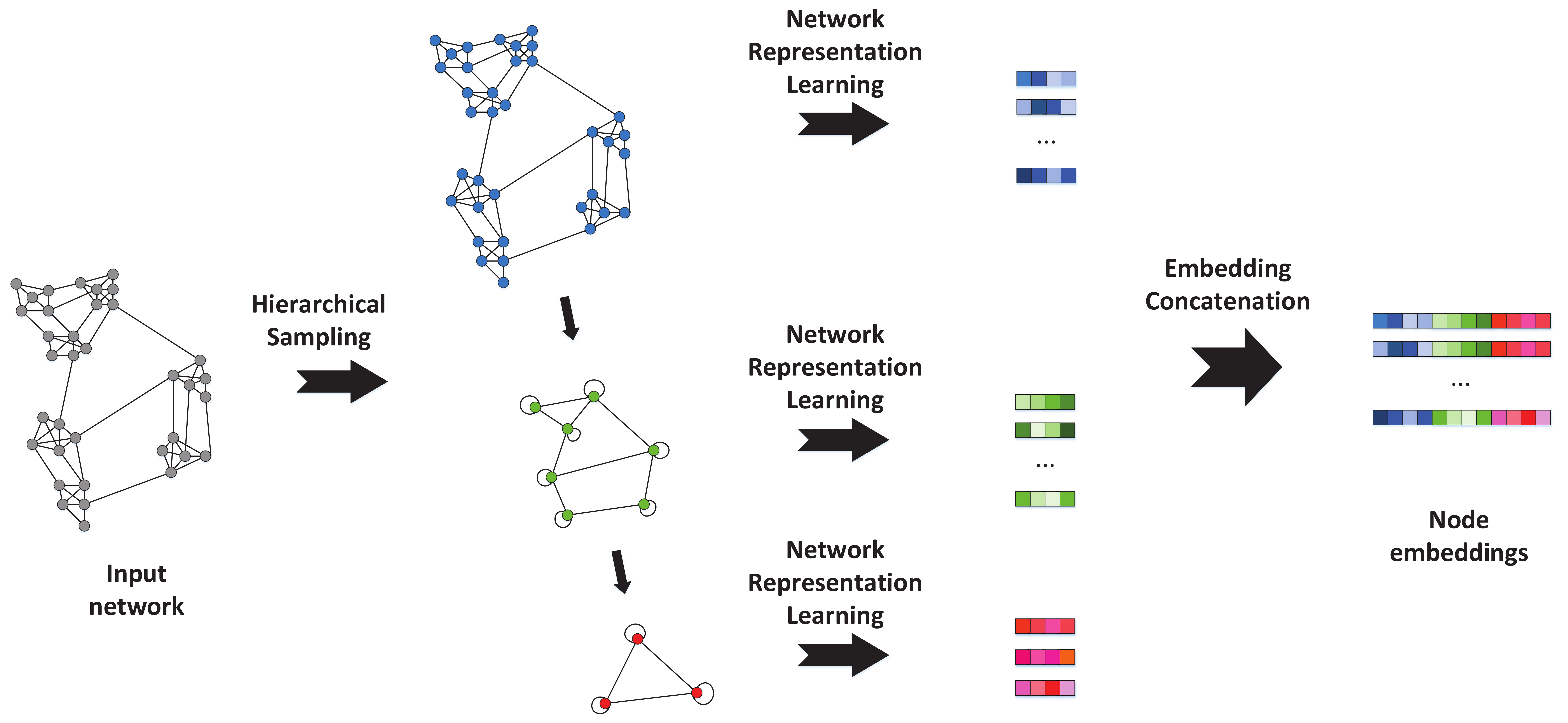}
	\caption{{\small The framework of HSRL.}}\label{figure3}
\end{figure*}

\subsection{Hierarchical Sampling}
Here we present the hierarchical sampling which is intended to compress a network into a series of compressed networks according to different compressing levels. Each compressed network reveals one of the hierarchical levels of global topology of the original network.

A community is one of the significant patterns of networks. Nodes in the same community are densely connected and nodes in different communities are sparsely connected. The relationship between nodes inside a community presents the local topological information of a network, while the relationship between communities reveals its global topology. It is worth noticing that in most large-scale networks, there are several natural organization levels - communities divide themselves into sub-communities - and thus communities with different hierarchical levels reveal the hierarchical topological information of original networks \cite{sales2007extracting, blondel2008fast}. Consequently, we compress a network into a new network based on communities by taking each community as a new node in the compressed network. Based on different hierarchical levels of communities, we can obtain a series of compressed networks which reveal the hierarchical global topological information of the input network.

The quality of the partitions obtained by community detection algorithms can be measured by the modularity of the partition \cite{clauset2004finding, blondel2008fast}. As a result, we can detect communities through optimizing the modularity of a network. As shown in Fig.\ref{figure2}, inspired by the Louvain method \cite{blondel2008fast}, hierarchical sampling compresses a network into a new network by implementing two phases: modularity optimization and node aggregation.

\textbf{Modularity optimization}. The first phase initializes each node in a network as a community and merges two connected nodes into one community if it can improve the modularity of the network. The implementation of community amalgamation will be repeated until a local maximum of the modularity is attained. 

\textbf{Node aggregation}. The second phase builds a new network whose nodes are the communities found in the previous phase. The weights of edges between new nodes are the sum of the weights of edges between nodes in the corresponding two communities. 

As shown in Algorithm \ref{algorithm1}, by recursively repeating the above two phases, hierarchical sampling obtains a series of compressed networks which reveal hierarchical global topology of the original network. 

\begin{algorithm}
	\renewcommand{\algorithmicrequire}{\textbf{Input:}}
	\renewcommand{\algorithmicensure}{\textbf{Output:}}
	\caption{Hierarchical Sampling}
	\label{algorithm1}
	\begin{algorithmic}[1]
		\REQUIRE network $G = (V, E)$, the largest \# compressed levels $K$
		\ENSURE a series of compressed networks $G^0$, $G^1$, ..., $G^K$
		\STATE $G^0 \leftarrow G$
		\FOR{$k \leq K$}
		\STATE $C^k \leftarrow ModularityOptimization(G^k)$
		\STATE $G^{k+1} \leftarrow NodeAggregation(C^k)$
		\STATE $k \leftarrow k + 1$
		\ENDFOR
		\STATE \textbf{return} $G^0$, $G^1$, ..., $G^K$
	\end{algorithmic}  
\end{algorithm}

\subsection{Representation Learning}
This section introduces representation learning on the compressed networks obtained by the previous section and concatenating the learned embeddings into node embeddings of the original network. We further provide a mathematical proof to demonstrate that HSRL embeds both local and global topological relationship of nodes in the original network into the learned embeddings.  

As shown in Fig.\ref{figure3}, we conduct representation learning on each compressed network. It is worth noticing that any NRL method can be used for this purpose. The embeddings of nodes in each compressed network are used to generate the final node embeddings of the original network. Particularly, the embedding $Z_i$ of node $v_i$ in the original network $G$ is the concatenation of the embeddings of hierarchical communities it involved in, as shown below.

\begin{equation}\label{equation2}
    Z_i = [Z_{c_i^0}^0, Z_{c_i^1}^1, ..., Z_{c_i^K}^K],
\end{equation}
\noindent where $c_i^k$ is the $k$-$th$ hierarchical community $v_i$ belongs to.

The node embeddings learned by the above representation learning process hold the following two Lemmas.

\newtheorem{myprop}{Lemma}
\begin{myprop}\label{lemma1}
    Nodes within the same hierarchical communities will get similar embeddings. The more the same hierarchical communities in which nodes involved, more similar embeddings they have.
\end{myprop}
From Eq.\ref{equation2}, it is easy to find that the above lemma holds. Lemma \ref{lemma1} shows that HSRL preserves the relationship between densely connected nodes in the original network. Therefore, HSRL is capable to preserve the local topological information of a network.

\begin{myprop}\label{lemma2}
    The cosine similarity between embedding $Z_i$ and embedding $Z_j$ is proportional to the sum of similarities of the embeddings between their hierarchical communities.
    \begin{equation}\label{equation3}
        sim(Z_i, Z_j) \propto \sum_{k=0}^Ksim(Z_{c_i^k}^k, Z_{c_j^k}^k).
    \end{equation}
\end{myprop}

\begin{proof}
    \begin{equation}
    \begin{aligned}
    	sim(Z_i, Z_j) = {} & \frac{Z_i \cdot Z_j}{\big|Z_i\big|\big|Z_j\big|} \\
    	\propto {} & [Z_{c_i^0}^0, Z_{c_i^1}^1, ..., Z_{c_i^K}^K] \cdot [Z_{c_j^0}^0, Z_{c_j^1}^1, ..., Z_{c_j^K}^K] \\
    	= {} & \sum_{k=0}^KZ_{c_i^k}^k \cdot Z_{c_j^k}^k \\
    	\propto {} & \sum_{k=0}^Ksim(Z_{c_i^k}^k, Z_{c_j^k}^k). \notag
    \end{aligned}
    \end{equation}
\end{proof}

From Lemma \ref{lemma2}, we know that two nodes will obtain similar embeddings if they are involved in similar hierarchical communities no matter the distance between them in the original network. The relationship between communities in different hierarchies is embedded in the embeddings of their involved nodes. Hence, HSRL can preserve the hierarchical global topological information of a network.

Finally, HSRL is presented in Algorithm \ref{algorithm2}.
\begin{algorithm}
	\renewcommand{\algorithmicrequire}{\textbf{Input:}}
	\renewcommand{\algorithmicensure}{\textbf{Output:}}
	\caption{HSRL}
	\label{algorithm2}
	\begin{algorithmic}[1]
		\REQUIRE network $G = (V, E)$, compressing levels $K$, NRL mapping function $f$
		\ENSURE node embeddings $Z$
		\STATE $G^0, G^1, ..., G^K \leftarrow HierarchicalSampling(G, K)$
		\FOR{$k \leq K$}
		\STATE $Z^k \leftarrow f(G^k)$
		\STATE $k \leftarrow k + 1$
		\ENDFOR
		\STATE $Z \leftarrow Concatenating(Z^0, Z^1, ..., Z^K)$
		\STATE \textbf{return} $Z$
	\end{algorithmic}  
\end{algorithm}

\section{Experiments}
In this section, five real-world datasets are used to evaluate the performance of HSRL on link prediction task. The source code is available at {\itshape https://github.com/fuguoji/HSRL}.

\subsection{Datasets}
We evaluate our method on various real-world datasets, including Movielens\footnote{https://movielens.org/}, MIT\cite{traud2012social}, DBLP\cite{sun2011pathsim}, Douban\cite{zheng2017recommendation}, and Yelp\footnote{https://www.yelp.com}. These datasets are commonly used in NRL field. The detailed statistics of datasets are shown in Table \ref{table1} and the brief descriptions of each dataset are presented as below.

\begin{itemize}
    \item \textbf{Movielens}: Movielens is a user-movie interest network which contains three types of nodes: {\itshape users}, {\itshape movies}, and {\itshape terms}.
    
    \item \textbf{DBLP}: DBLP is a bibliographic network in computer science collected from four research areas: database, data mining, machine learning, and information retrieval. Nodes in the network including {\itshape authors}, {\itshape papers}, {\itshape venues}, and {\itshape terms}.
    
    \item \textbf{MIT}: MIT is a Facebook friendship network at one hundred American colleges and universities at a single point in time. It contains a single type of nodes, {\itshape users}.
    
    \item \textbf{Douban}: Douban is a user-movie interest network collected from a user review website Douban in China. The network contains four types of nodes including {\itshape users}, {\itshape movies}, {\itshape actors}, and {\itshape directors}.
    
    \item \textbf{Yelp}: Yelp is a user-business network collected from a website Yelp in America. It contains four types of nodes including {\itshape users}, {\itshape businesses}, {\itshape locations}, and {\itshape business categories}.
\end{itemize}

\begin{table}
	\centering
	\caption{{\small Statistics of five datasets.}}\label{table1}
	\begin{tabular}{cccc}
		\toprule
		\textbf{Datasets} & \# Nodes & \# Edges & Network Types \\
		\midrule
		Movielens & 1332 & 2592 & User-movie \\
		DBLP & 37791 & 170794 & Bibliography \\
		MIT & 6402 & 251230 & Friendship \\
		Douban & 13786 & 214392 & User-movie\\
		Yelp & 28759 & 247698 & User-business\\
		\bottomrule
	\end{tabular}
\end{table}

\subsection{Baselines}
We compare our method with four start-of-the-art algorithms, which are introduced below. 
\begin{itemize}
    \item \textbf{DeepWalk}: DeepWalk is a random walk based NRL method. It conducts random walks on each node to sample node sequences from a network and uses the Skip-Gram model to learn node embeddings by treating node sequences as sentences and nodes as words.
    
    \item \textbf{node2vec}: node2vec is a biased random walk based method that provides a trade-off between DFS and BFS when employing random walks on nodes. Then the Skip-Gram model is used to learn node embeddings based on the sampling node sequences.
    
    \item \textbf{LINE}: LINE defines the first-order and second-order proximities to measure the similarity between nodes, and learns node embeddings by preserving the aforementioned proximities of nodes in the embedding space.
    
    \item \textbf{HARP}: HARP recursively uses two collapsing schemes, edge collapsing and star collapsing, to compress an input network into a series of small networks. Starting from the smallest compressed network, it recursively conducts a NRL method to learns node embeddings in each network using node embeddings of the previous level network as initialization.
\end{itemize}

\subsection{Parameter Settings}
Here we discuss the parameter setting for our method and baselines:
\begin{itemize}
    \item \textbf{DeepWalk}, \textbf{HARP(DW)}, and \textbf{HSRL(DW)}. The number of random walks $t$, the length of each walk $l$, window size of Skip-Gram model $w$, representation size $d$, and learning rate $\eta$ for DeepWalk, HARP(DW), and HSRL(DW) are set as $t = 10$, $l = 40$, $w = 5$, $d = 64$, $\eta = 0.025$. The largest number of compressing levels for HSRL(DW) $K$ is set as 3.
    
    \item \textbf{node2vec}, \textbf{HARP(N2V)}, and \textbf{HSRL(N2V)}. The parameter setting for node2vec, HARP(N2V), and HSRL(N2V) is $t = 10$, $l = 40$, $w = 5$, $d = 64$, $\eta = 0.025$. The largest number of compressing levels for HSRL(N2V) $K$ is set as 3.
    
    \item \textbf{LINE}, \textbf{HARP(LINE)}, and \textbf{HSRL(LINE)}. The number of negative sampling $\lambda$, learning rate $\eta$, and representation size $d$ for LINE, HARP(LINE), and HSRL(LINE) are set as $\lambda = 5$, $\eta = 0.025$, $d = 64$. The largest number of compressing levels for HSRL(LINE) $K$ is set as 3.
\end{itemize}
\subsection{Hierarchical Sampling for Networks}

\begin{figure*}
	\centering
	\subfigure[{\small Movielens}]{
		\includegraphics[scale=0.3]{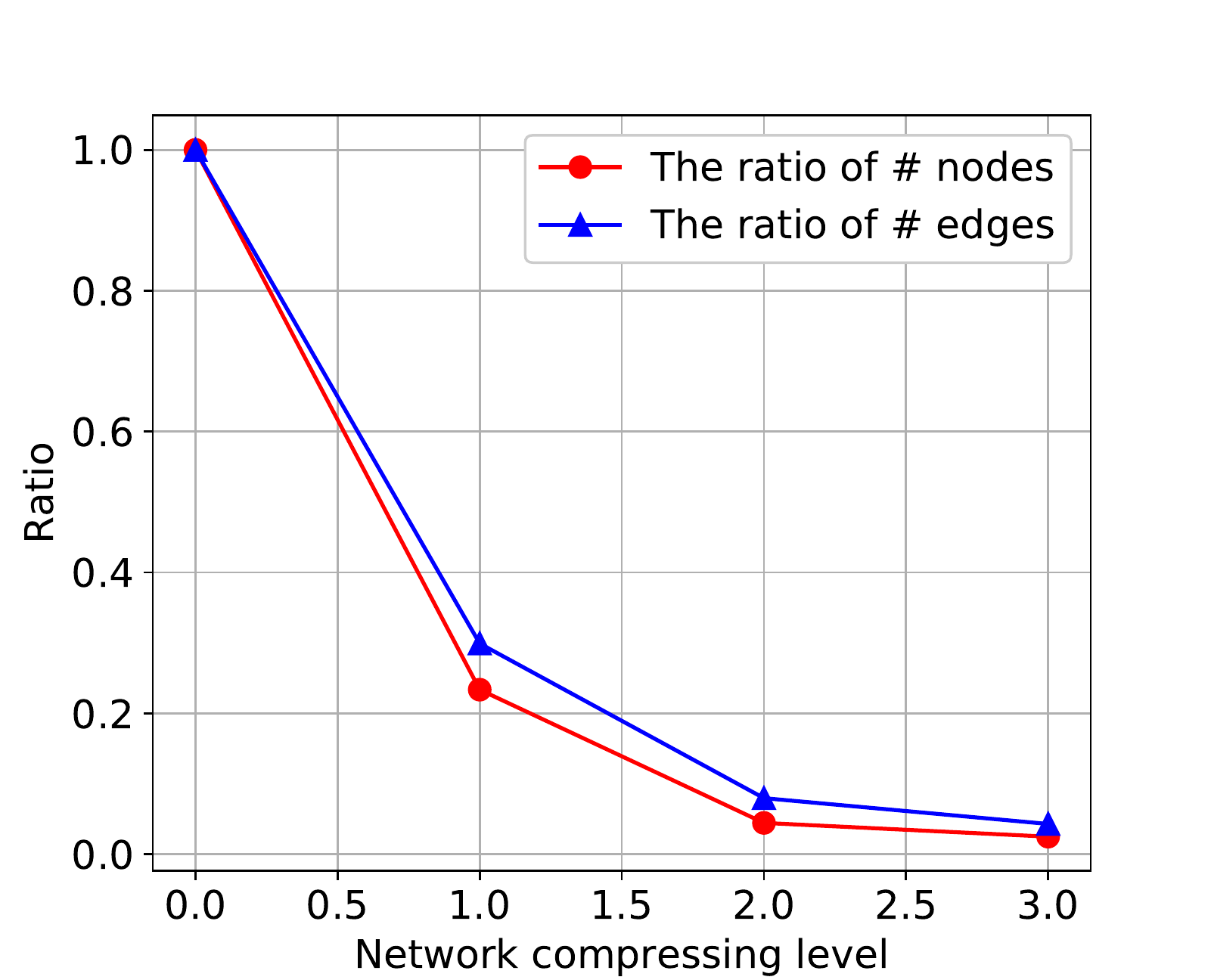}
	}
	\subfigure[{\small DBLP}]{
		\includegraphics[scale=0.3]{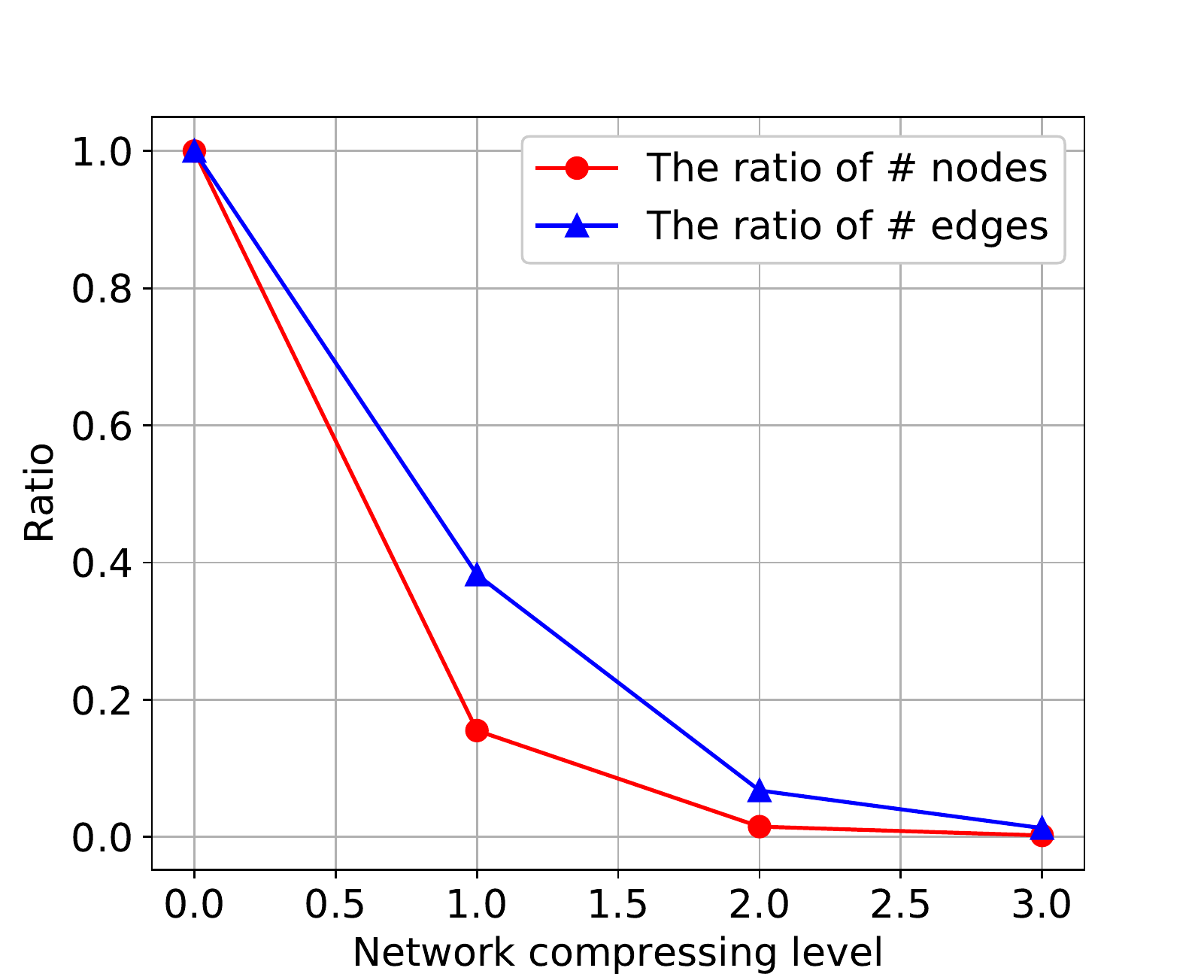}
	}
	\subfigure[{\small MIT}]{
		\includegraphics[scale=0.3]{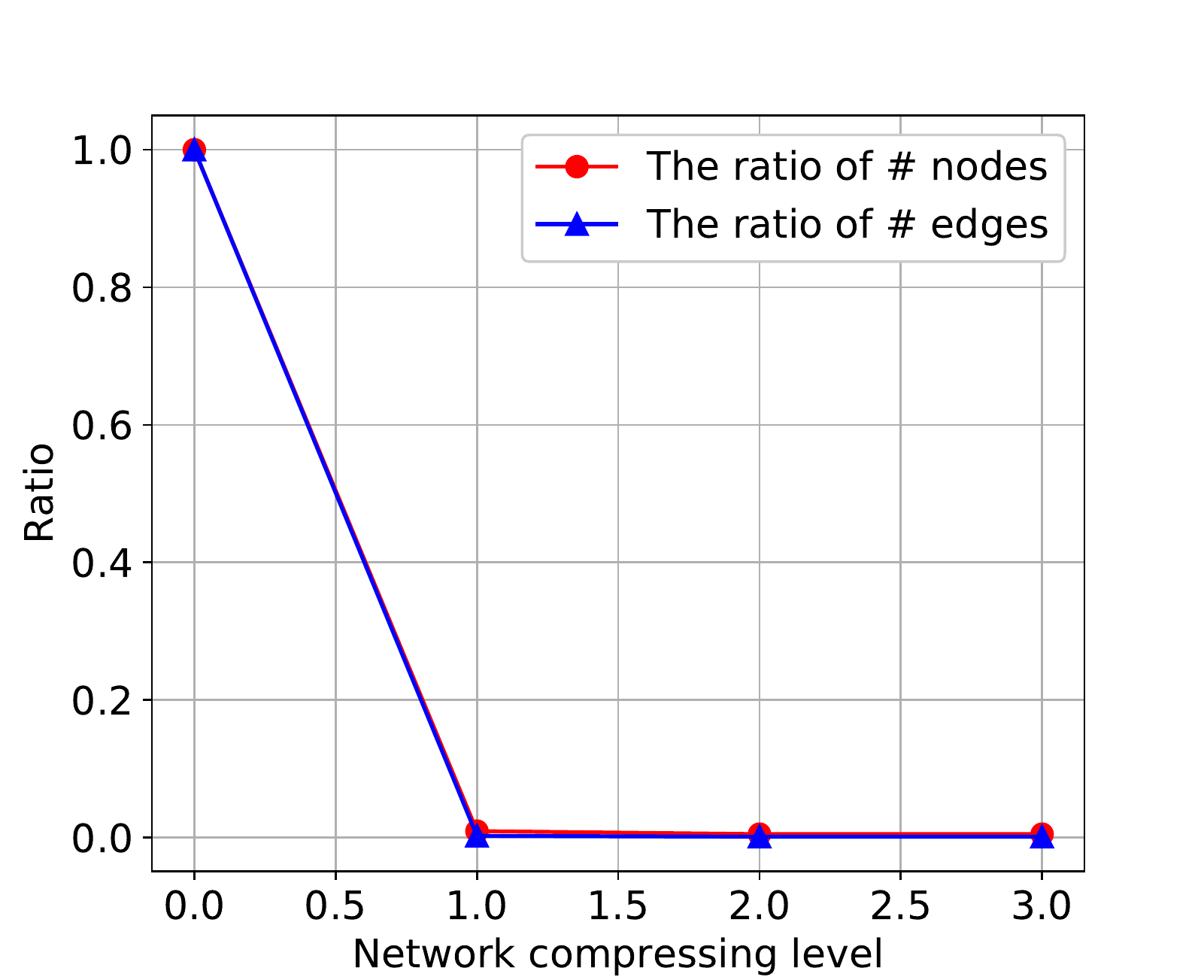}
	}
	\subfigure[{\small Douban}]{
		\includegraphics[scale=0.3]{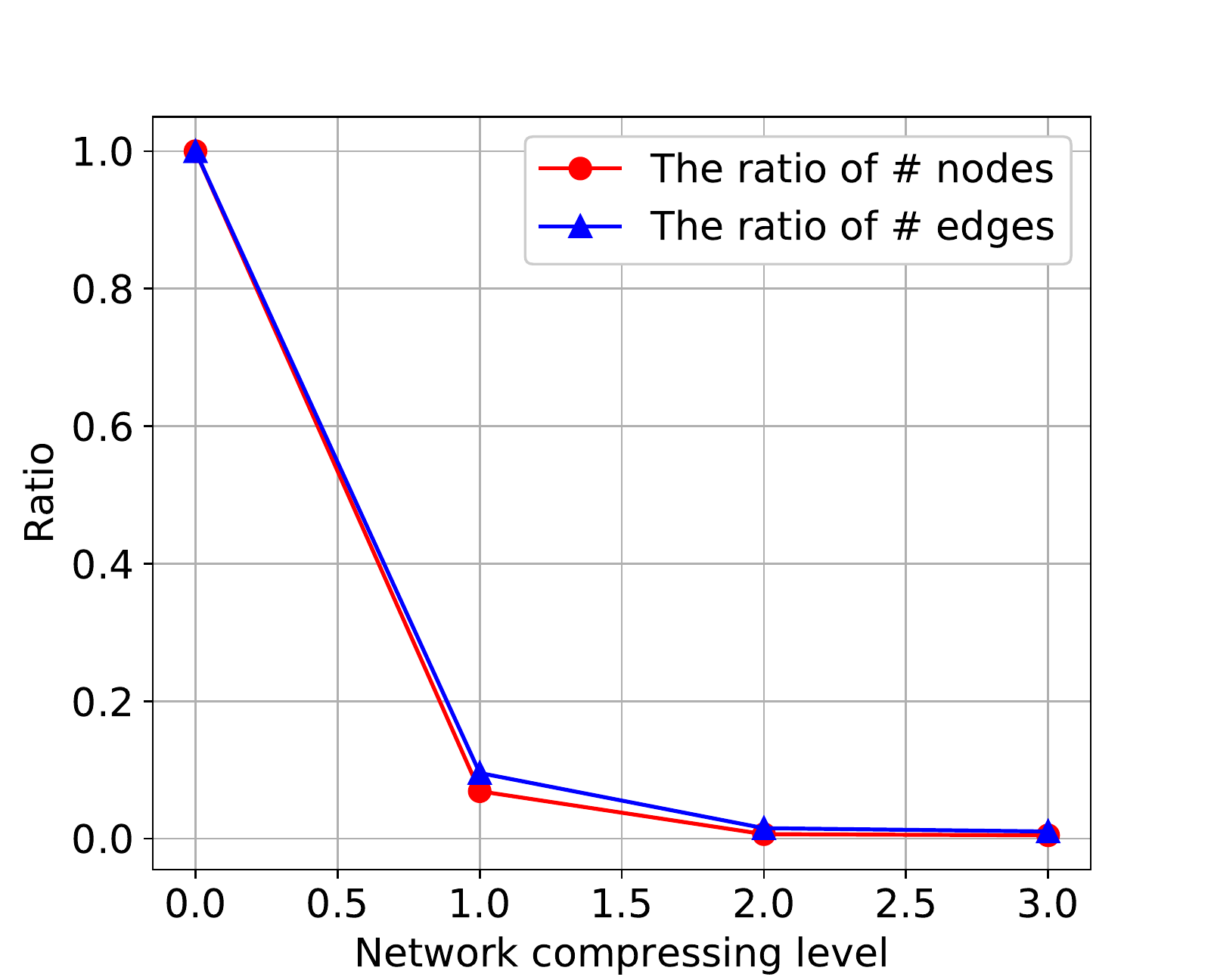}
	}
	\subfigure[{\small Yelp}]{
		\includegraphics[scale=0.3]{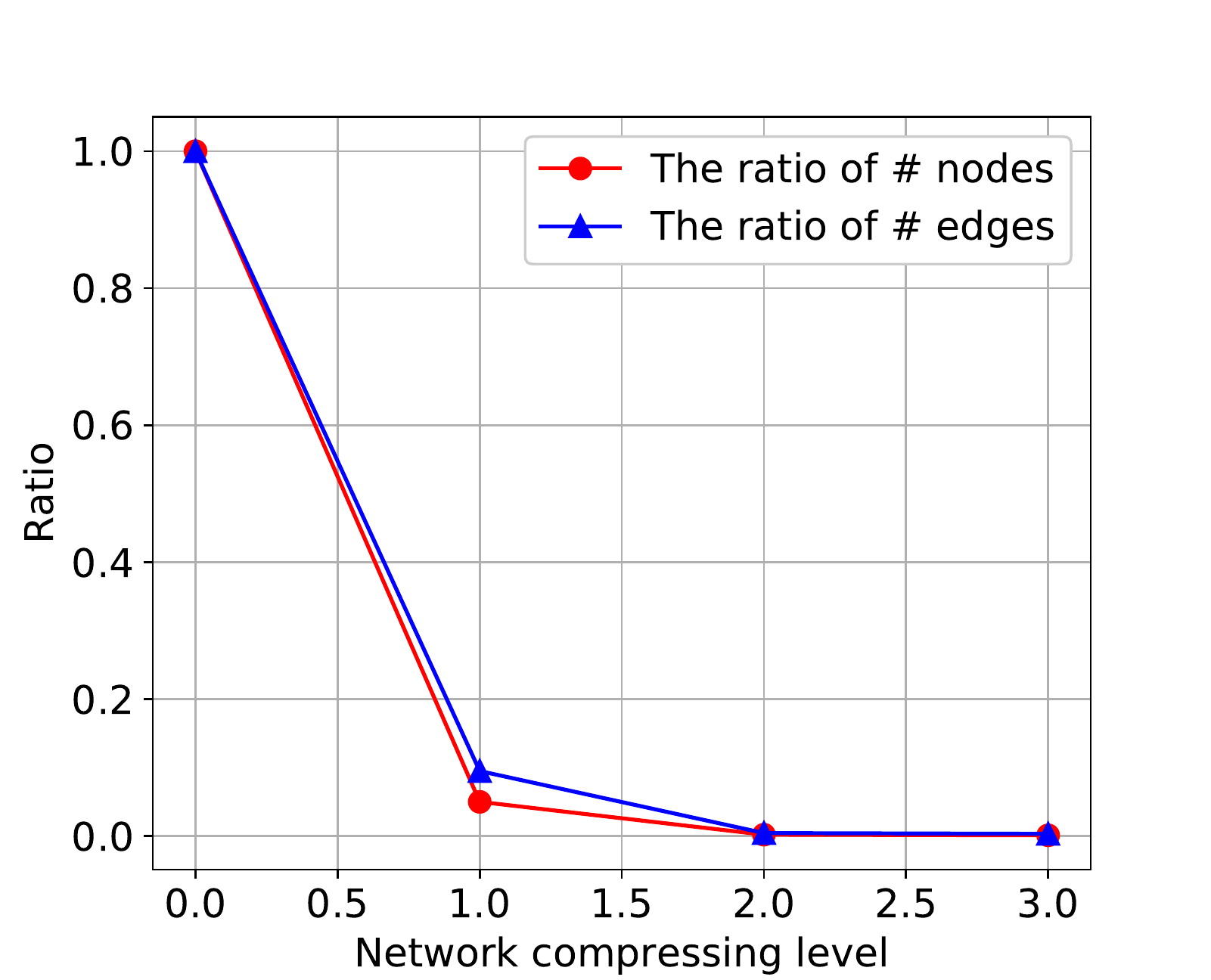}
	}
	\caption{{\small The compressing ratio of nodes$/$edges of compressed networks to the original network.}}\label{figure4}
\end{figure*}

We firstly discuss the results of hierarchical sampling on testing networks. Fig.\ref{figure4} presents the network compressing results by the hierarchical sampling on five datasets. As shown in Fig.\ref{figure4}, the number of nodes and edges of compressed networks drastically decrease as the compressing process continues and finally becomes stable when the compressing level is large than 3. Therefore, in the following link prediction tasks, we set the largest number of compressing level as 3.

\begin{figure}
	\centering
	\subfigure[{\small Input}]{
		\includegraphics[scale=0.08]{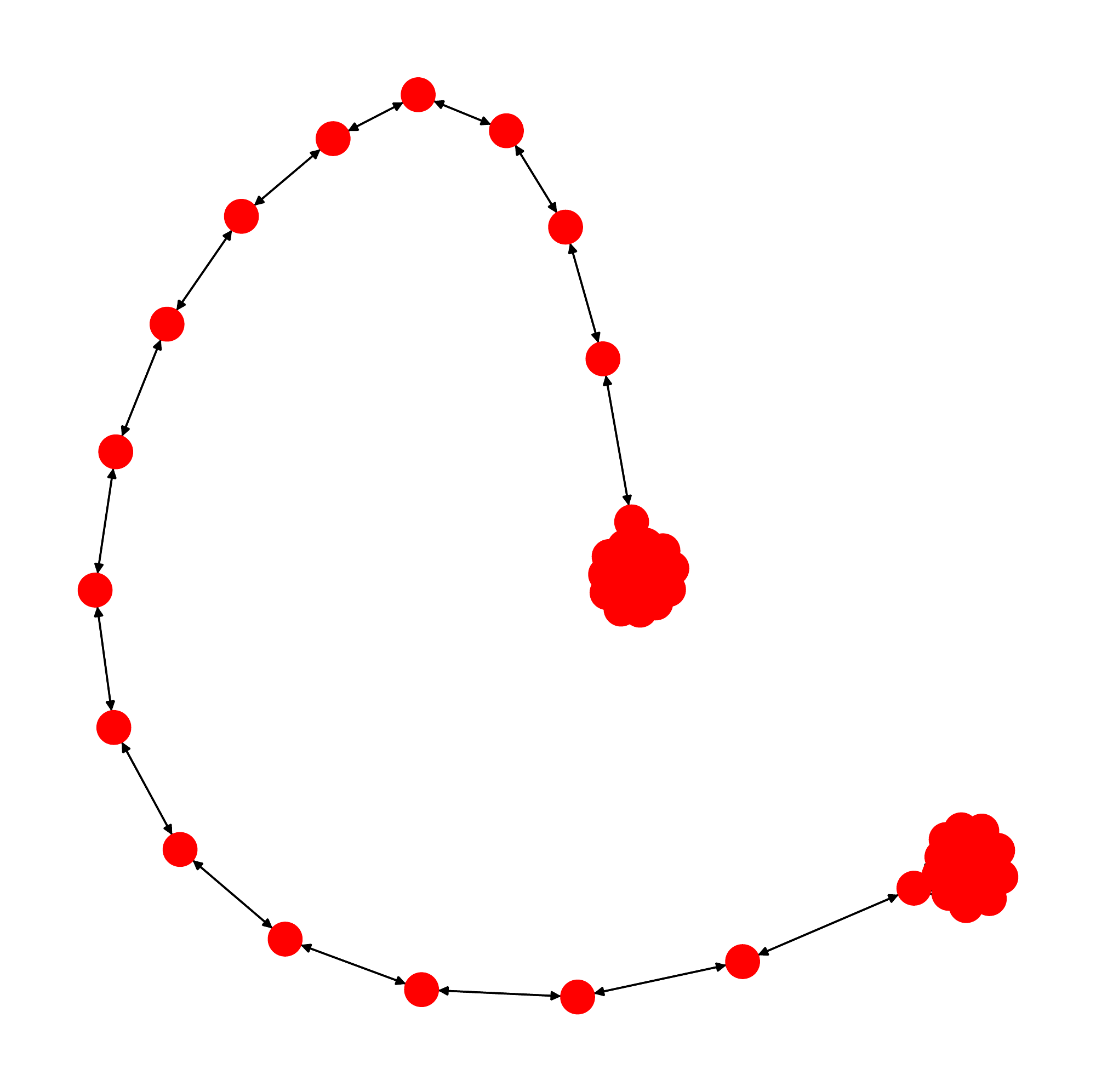}
	}
	\subfigure[{\small Level 1}]{
		\includegraphics[scale=0.08]{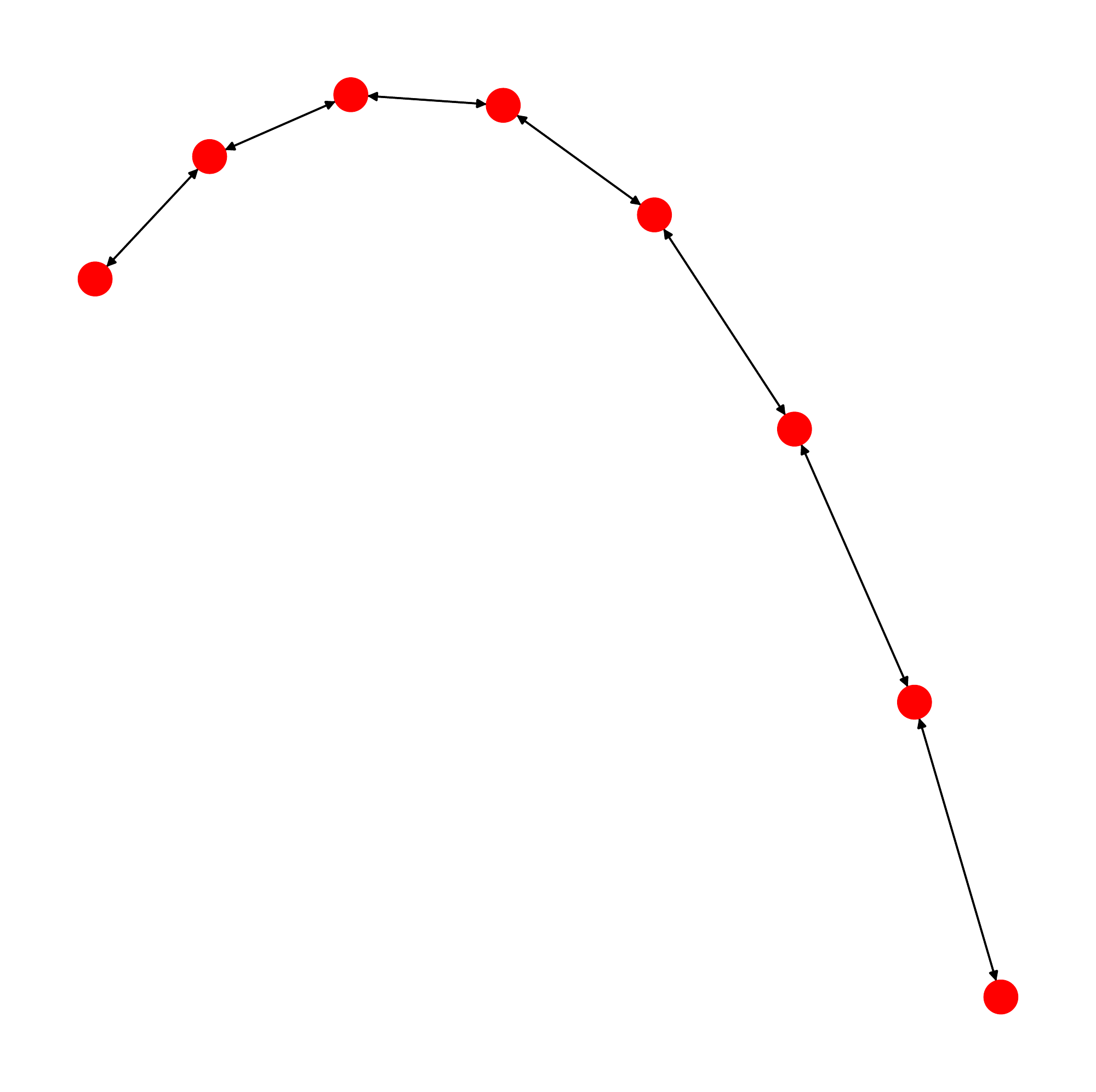}
	}
	\subfigure[{\small Level 2}]{
		\includegraphics[scale=0.08]{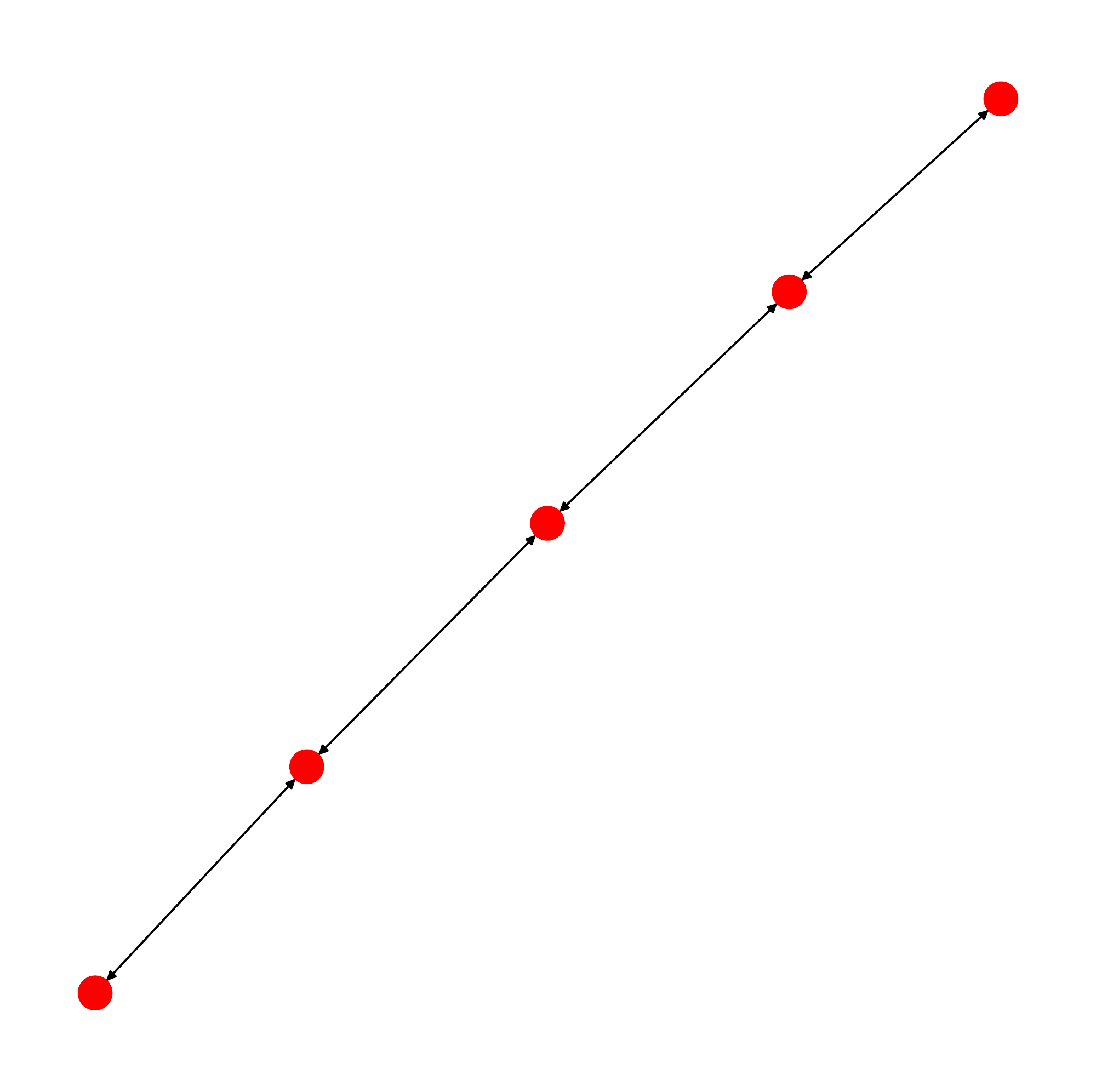}
	}
	\subfigure[{\small Level 3}]{
		\includegraphics[scale=0.08]{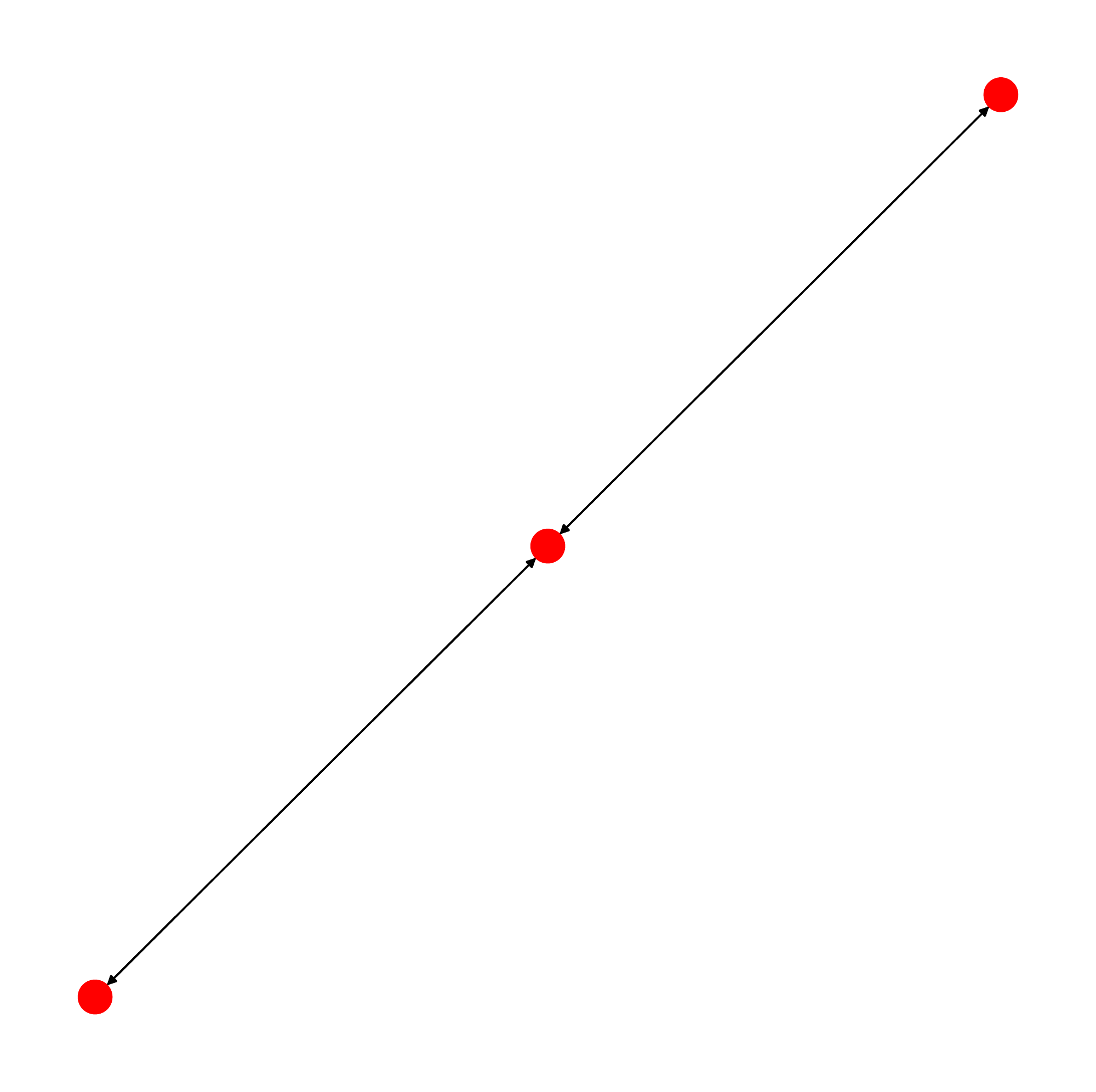}
	}
	\subfigure[{\small Input}]{
		\includegraphics[scale=0.08]{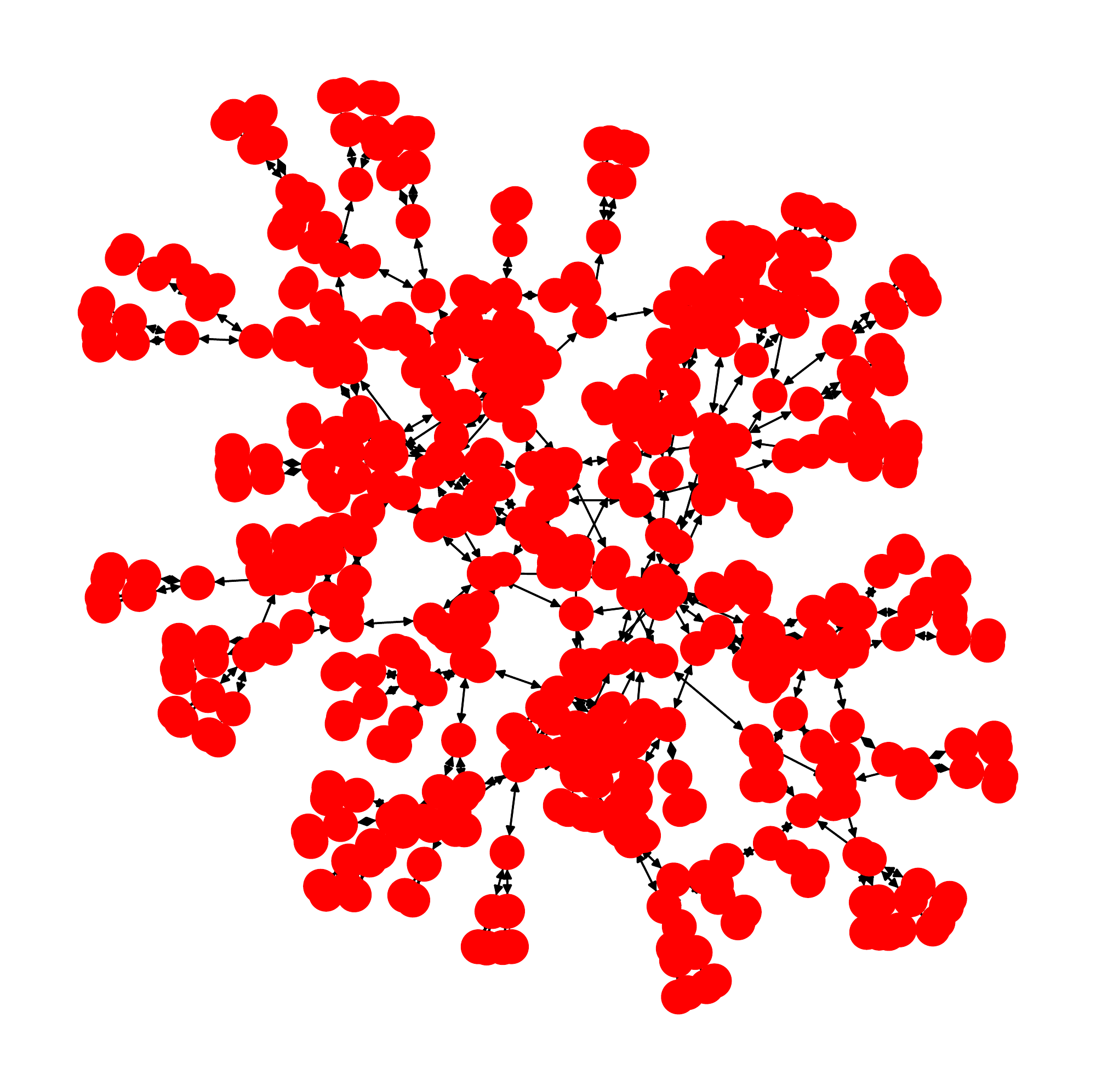}
	}
	\subfigure[{\small Level 1}]{
		\includegraphics[scale=0.08]{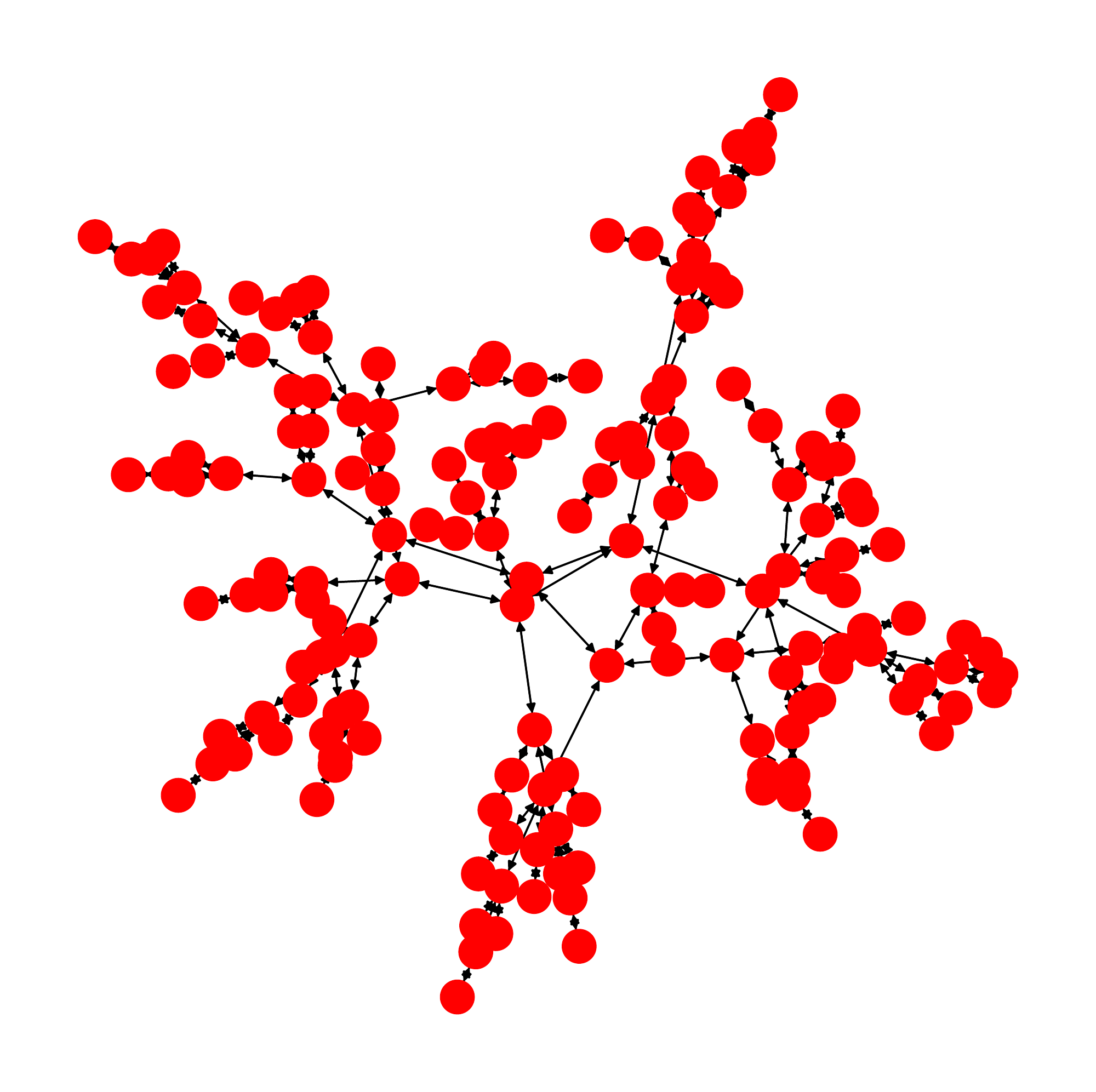}
	}
	\subfigure[{\small Level 2}]{
		\includegraphics[scale=0.08]{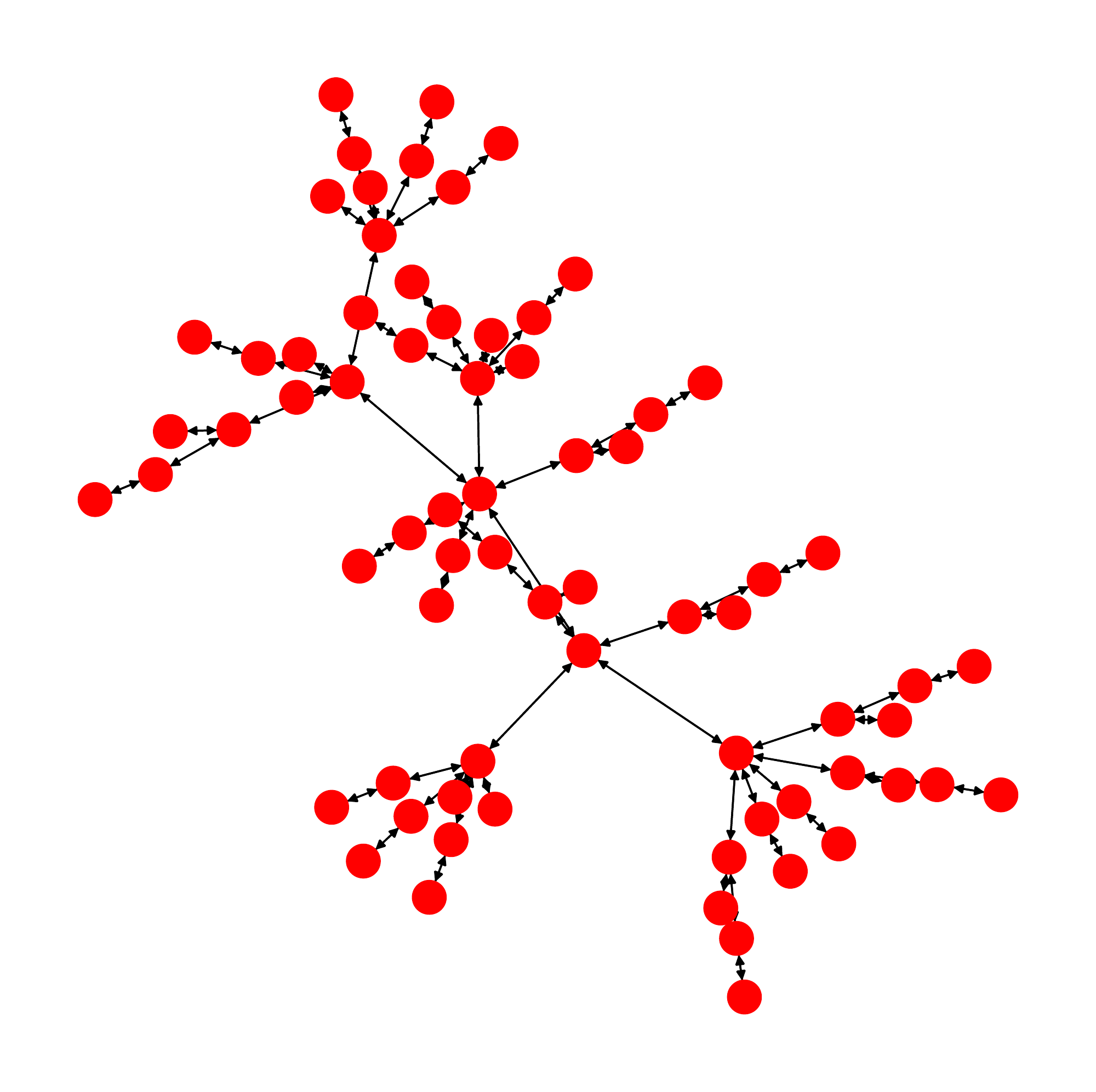}
	}
	\subfigure[{\small Level 3}]{
		\includegraphics[scale=0.08]{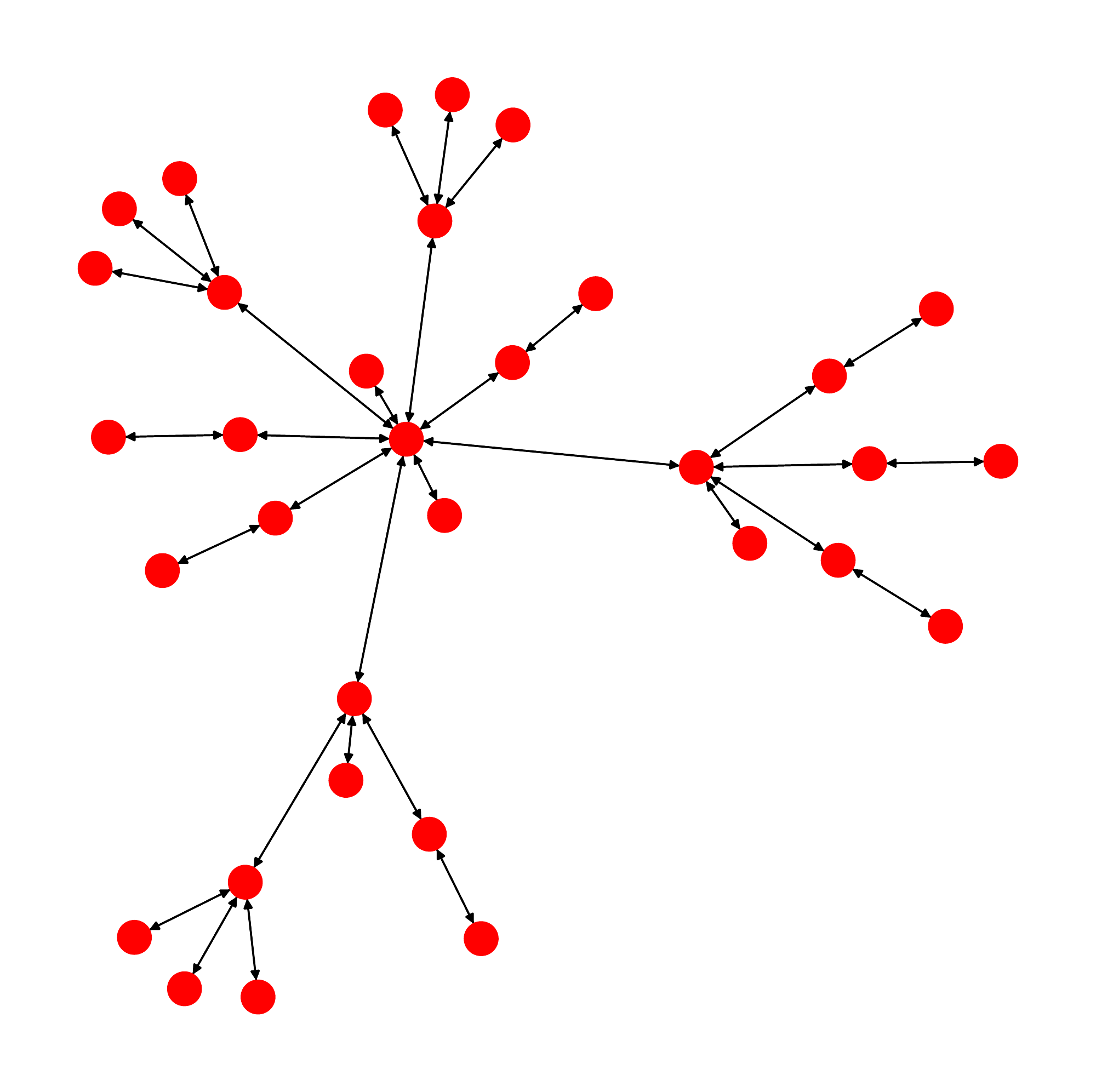}
	}
	\subfigure[{\small Input}]{
		\includegraphics[scale=0.08]{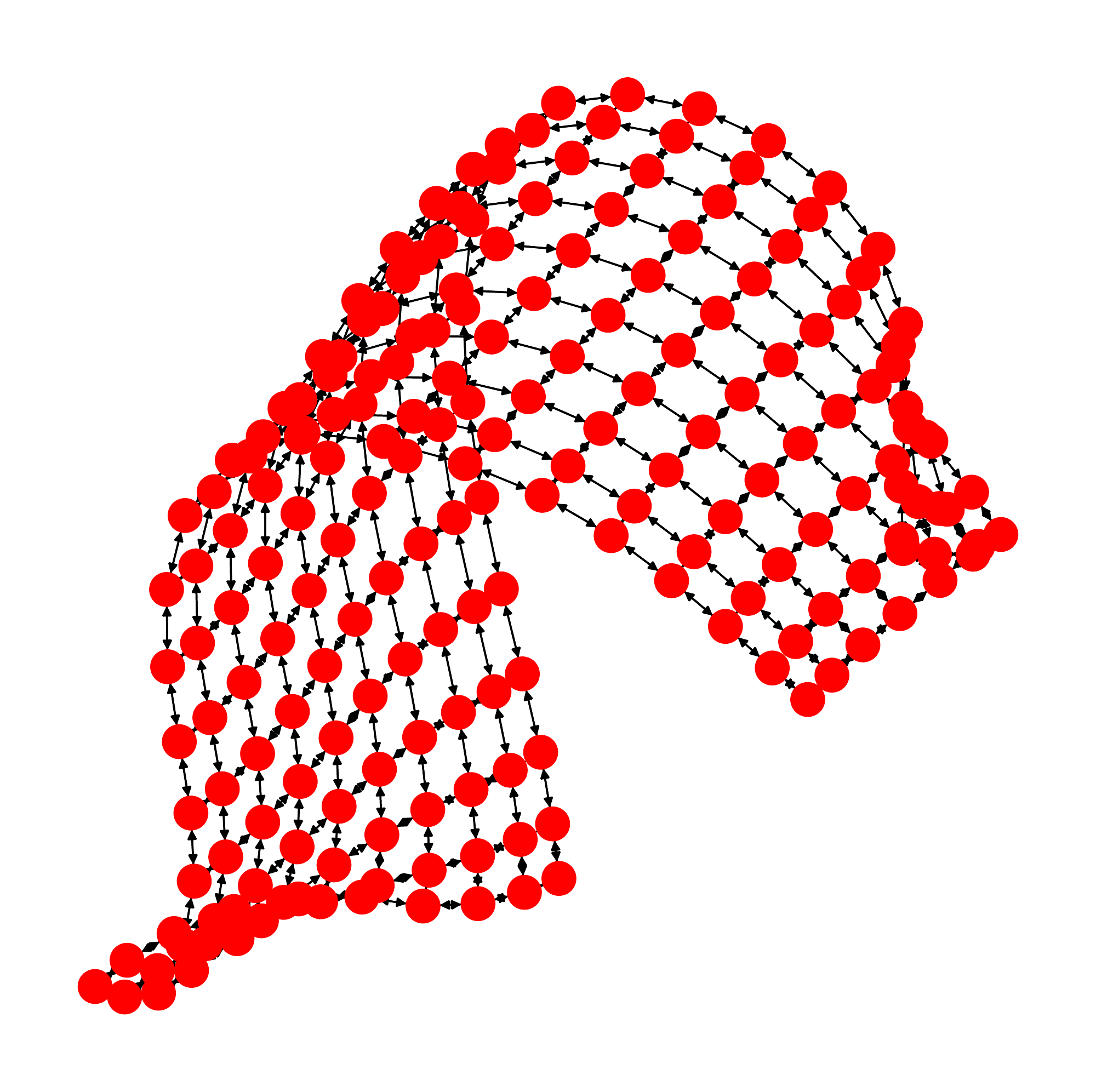}
	}
	\subfigure[{\small Level 1}]{
		\includegraphics[scale=0.08]{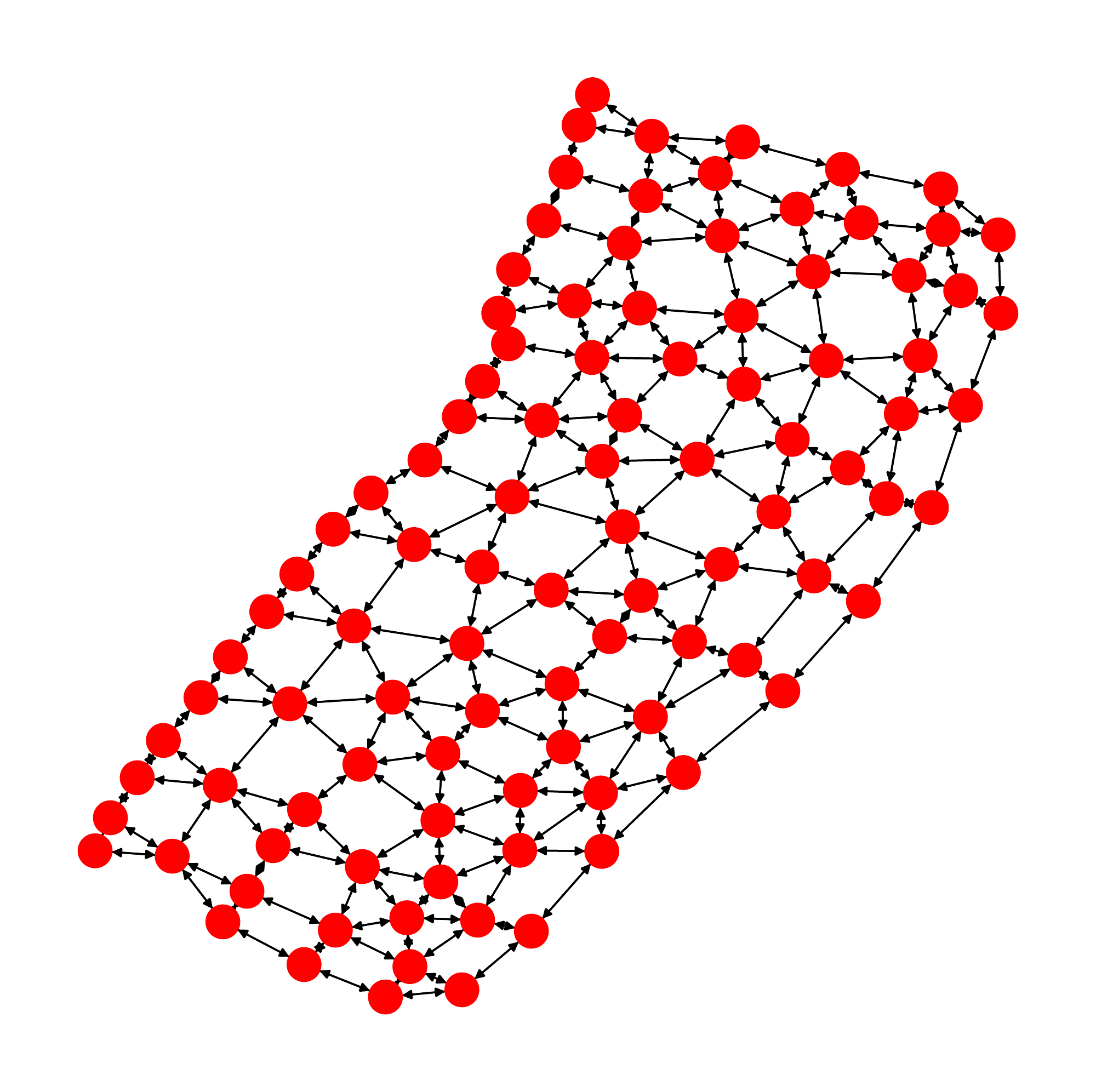}
	}
	\subfigure[{\small Level 2}]{
		\includegraphics[scale=0.08]{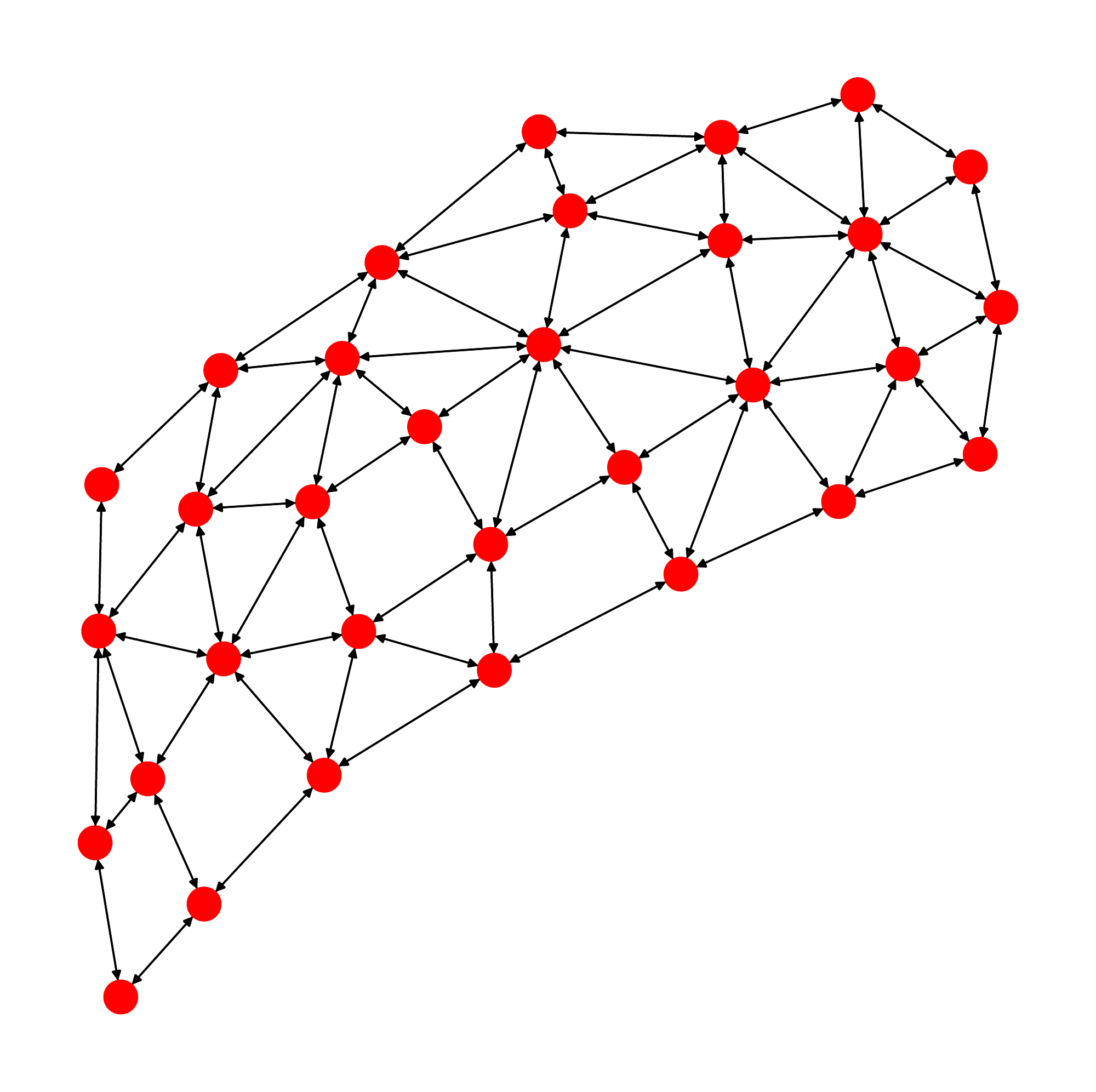}
	}
	\subfigure[{\small Level 3}]{
		\includegraphics[scale=0.08]{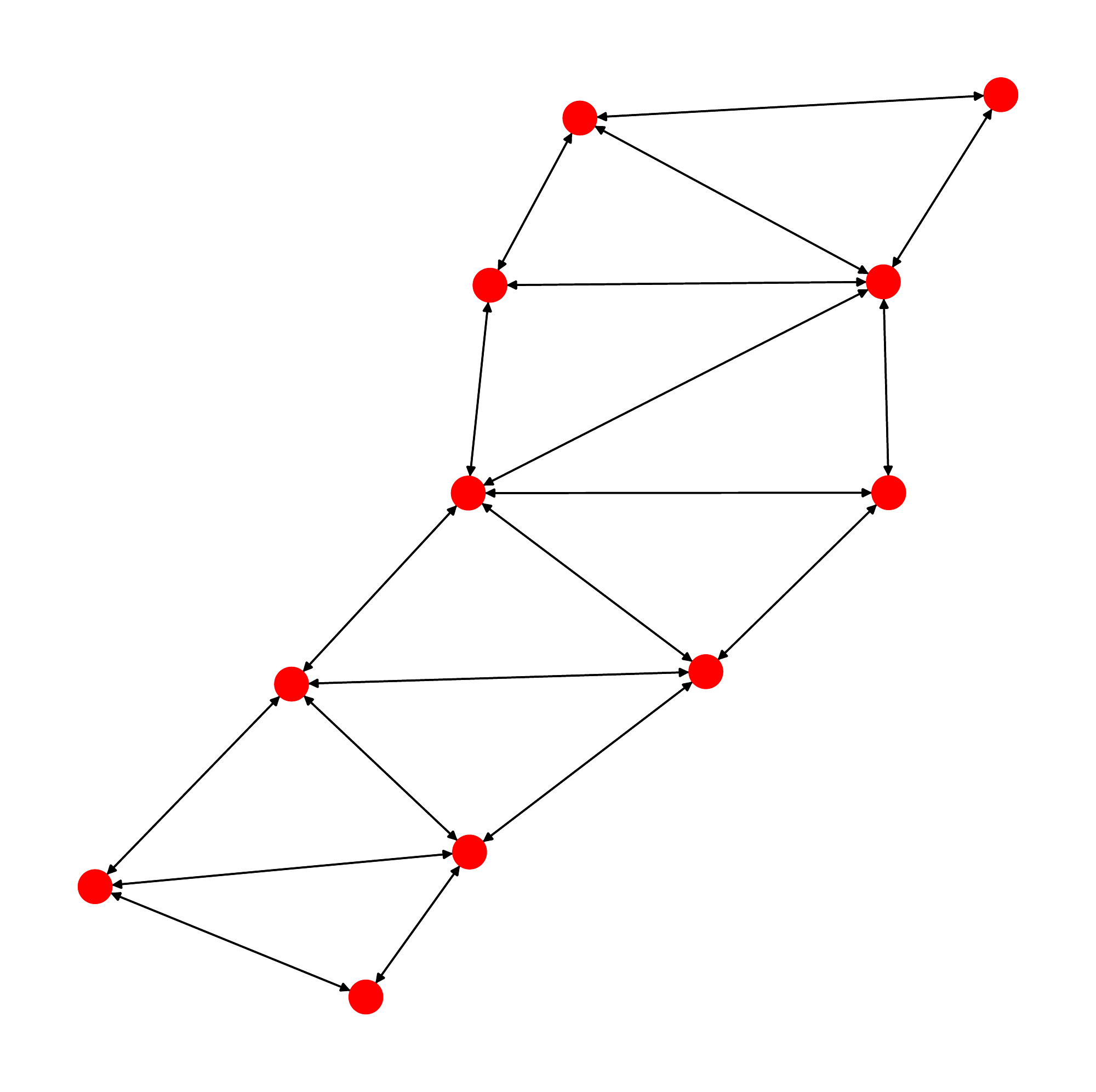}
	}
	\caption{{\small Examples of hierarchical sampling.}}\label{figure5}
\end{figure}

As shown in Fig.\ref{figure5}, we present three networks, Fig.\ref{figure5}(a), Fig.\ref{figure5}(e), and Fig.\ref{figure5}(i), as the intuitive examples to illustrate how hierarchical sampling works. 

The network in Fig.\ref{figure5}(a) contains two dense communities which are merged into a new node in the following compressed networks respectively. Therefore in the compressed networks, the local topological information of original networks are preserved by considering densely connected nodes in the same community as a whole. Meanwhile, the compressed networks, Fig.\ref{figure5}(b), Fig.\ref{figure5}(c), and Fig.\ref{figure5}(d), reveal the hierarchical topological information of the input network. For a balanced tree network and a grid network as shown in Fig.\ref{figure5}(e) and Fig.\ref{figure5}(i), their hierarchical topologies can be revealed by their compressed networks as well. The results of Fig.\ref{figure5} show that the network compressing strategy of HSRL works well on different types of networks.

\subsection{Link Prediction}
We conduct link prediction tasks to evaluate the performance of our method on five real-world datasets. Specifically, we prediction the link between nodes based on the cosine similarity of their embeddings. The evaluation metric used in this task is AUC. Higher AUC indicates the better performance of NRL methods.

We randomly split the edges of a network into 80\% edges as the training set and the left 20\% as the testing set. Each experiment is independently implemented for 20 times and the average performances on the testing set are reported in Table \ref{table2}.

We summarize the observations from Table.\ref{table2} as following:
\begin{itemize}
	\item HSRL significantly outperforms all baselines on all datasets. For the small and sparse network movielens, the improvements of HSRL(DW), HSRL(N2V), and HSRL(LINE) are 3.6\%, 2.5\%, and 16.6\% respectively. For the dense network MIT, HSRL(DW), HSRL(N2V), and HSRL(LINE) outperform the baselines by 2.9\%, 8.5\%, and 0.5\%. For three large networks, DBLP, Yelp, and Douban, the improvement of HSRL is striking: the improvements of HSRL(DW), HSRL(N2V), and HSRL(LINE) are 6.3\%, 19.9\%, 3.5\% for DBLP, 6.5\%, 16.8\%, 5.9\% for Yelp, and 18.4\%, 30.5\%, 17.5\% for Douban. 
	
	\item The results of HARP on movielens, DBLP, Yelp, and Douban are worse than the original NRL methods. Moreover, the performance of HARP(LINE) is drastically worse than LINE. It only works better than DeepWalk, node2vec on MIT which is a small and dense network. The compressed networks generated by HARP on a network could not reveal its global topologies. Hence, using node embeddings of compressed networks as initialization could mislead the NRL methods to a bad local minimum. Such an issue could occur especially when the input network is large-scale and the objective function of the NRL method is highly non-convex, e.g., LINE.
	
	\item The improvements of HSRL on DBLP, Yelp, and Duban are larger than that on Movielens and MIT. It demonstrates that HSRL works much better than baselines on large-scale networks. 
\end{itemize}
\begin{table}
    \caption{{\small AUC of link prediction.}}\label{table2}
    \begin{tabular}{cccccc}
        \toprule
         \textbf{Algorithm} & \multicolumn{5}{c}{\textbf{Dataset}} \\
          & Movielens & DBLP & MIT & Yelp & Douban \\
         \midrule
         DeepWalk & 0.847 & 0.794 & 0.899 & 0.842 & 0.687 \\
         HARP(DW) & 0.817 & 0.659 & 0.902 & 0.743 & 0.559 \\
         HSRL(DW) & \textbf{0.879}$^\dag$ & \textbf{0.847}$^\dag$ & \textbf{0.926}$^\dag$ & \textbf{0.901}$^\dag$ & \textbf{0.842}$^\dag$ \\
         Gain of HSRL(\%) & \textbf{3.6} & \textbf{6.3} & \textbf{2.9} & \textbf{6.5} & \textbf{18.4} \\
         \midrule
         node2vec & 0.843 & 0.673 & 0.843 & 0.742 & 0.569 \\
         HARP(N2V) & 0.828 & 0.647 & 0.879 & 0.708 & 0.552 \\
         HSRL(N2V) & \textbf{0.865}$^\dag$ & \textbf{0.840}$^\dag$ & \textbf{0.921}$^\dag$ & \textbf{0.892}$^\dag$ & \textbf{0.819}$^\dag$ \\
         Gain of HSRL(\%) & \textbf{2.5} & \textbf{19.9} & \textbf{8.5} & \textbf{16.8} & \textbf{30.5} \\
         \midrule
         LINE & 0.613 & 0.641 & 0.814 & 0.752 & 0.624 \\
         HARP(LINE) & 0.220 & 0.387 & 0.702 & 0.306 & 0.399 \\
         HSRL(LINE) & \textbf{0.735}$^\dag$ & \textbf{0.664}$^\dag$ & \textbf{0.819} & \textbf{0.799}$^\dag$ & \textbf{0.756}$^\dag$ \\
         Gain of HSRL(\%) & \textbf{16.6} & \textbf{3.5} & \textbf{0.5} & \textbf{5.9} & \textbf{17.5} \\
         \bottomrule
    \end{tabular}
	\smallskip
	
	{\small $^\dag$ denotes the performance of HSRL is significantly better than the other peers according to the Wilcoxon’s rank sum test at a 0.05 significance level.}
\end{table}

\subsection{Parameter Sensitivity Analysis}
We conduct link prediction task on DBLP to study the parameter sensitivity of HSRL. Without loss of generality, we used DeepWalk to learn node embeddings for each compressed network. Fig.\ref{figure6} shows using 80\% edges as training set and the left as testing set, the link prediction performance (AUC) as a function of one of the chosen parameters when fixing the others.

When fixing the largest number of compressed level to 3, Fig.\ref{figure6}(a) shows the AUC of link prediction drastically improves as the number of embedding dimension $d$ increases and finally becomes stable when $d$ is larger than 32. When $d$ is small, it is inadequate to embody rich information of networks. However, when $d$ is large enough to embody all original network information, increasing $d$ will not improve the performance of link prediction.

Fig.\ref{figure6}(b) shows that the impact of the largest number of network compressing level $K$ on the performance of link prediction by fixing the representation size $d$ to 64. As we increase $K$, the AUC of link prediction drastically improves. It demonstrates that the hierarchical topologies help to capture the potential relationship between nodes (even they are far away from each other) in a network. When $K$ is larger than 3, the performance of link prediction becomes stable. It is reasonable since the DBLP network could not be compressed further after level 3 as shown in Fig.\ref{figure4}(b).

\begin{figure}
	\centering
	\subfigure[{representation size $d$}]{
		\includegraphics[scale=0.19]{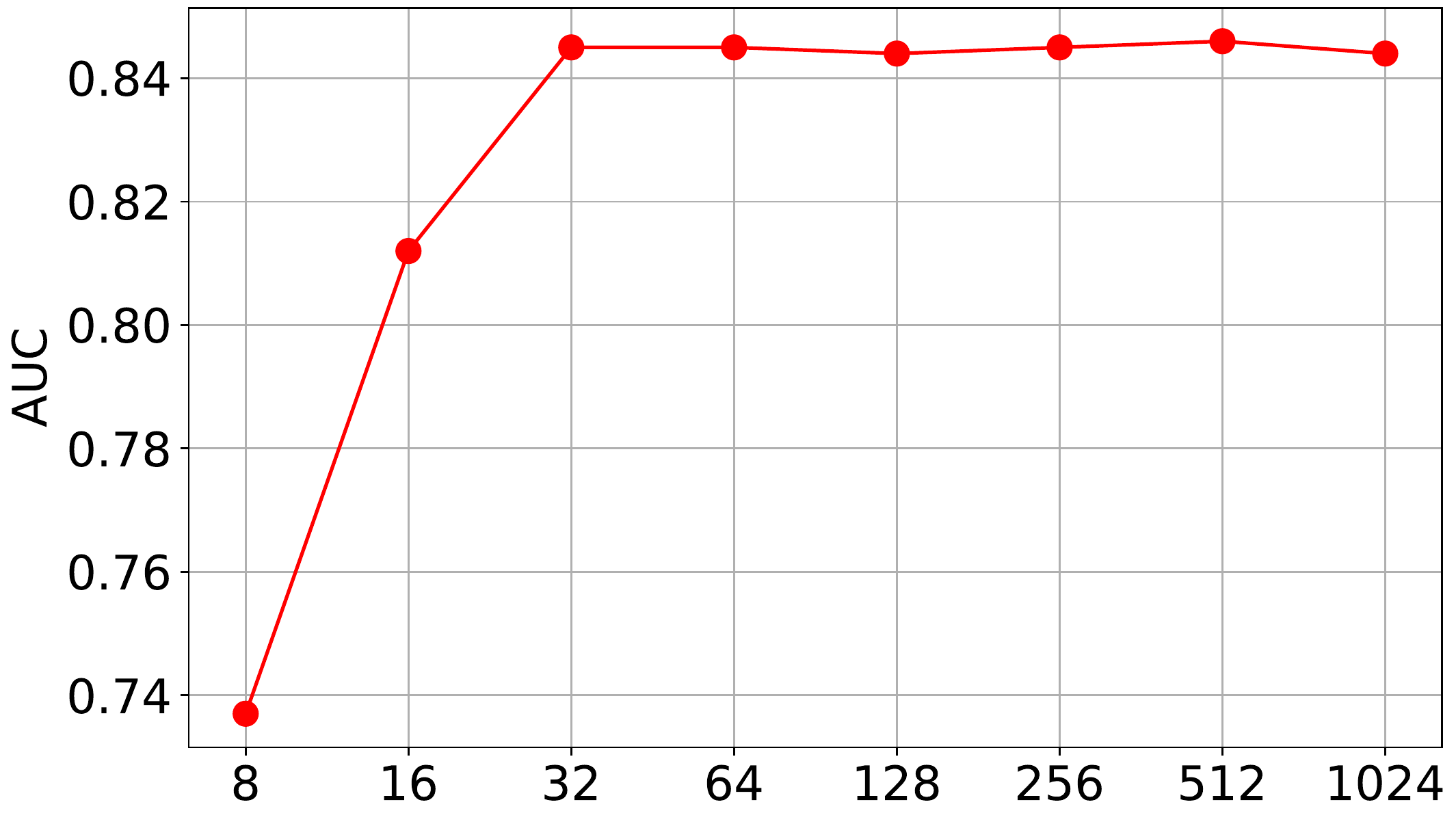}
	}
	\subfigure[network compressing level $K$]{
		\includegraphics[scale=0.19]{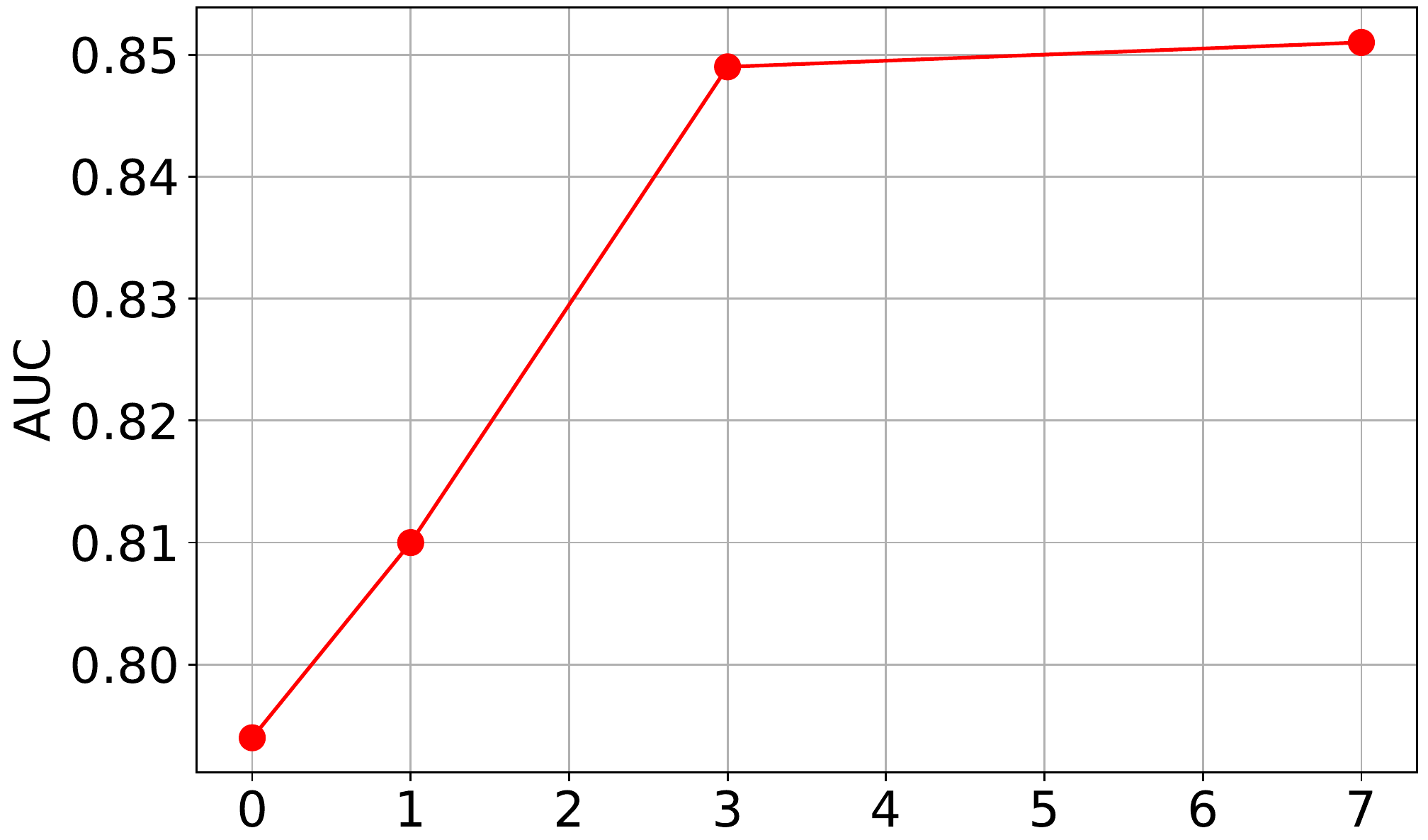}
	}
	\caption{{\small Parameter sensitivity on link prediction.}}\label{figure6}
\end{figure}

\subsection{Running Time}
Fig.\ref{figure6} shows the actual running time of all NRL methods on five testing networks. All experiments are conducted on a single machine with 32GB memory, 16 CPU cores at 3.2 GHZ. The results show that the actual running time of HSRL is at most three times higher than others. The running time of HSRL is linear to the corresponding baselines as the input networks growing. Moreover, the running time of HSRL can be reduced by parallelizing the training processes on all compressed networks.
\begin{figure}
	\centering
	\includegraphics[scale=0.38]{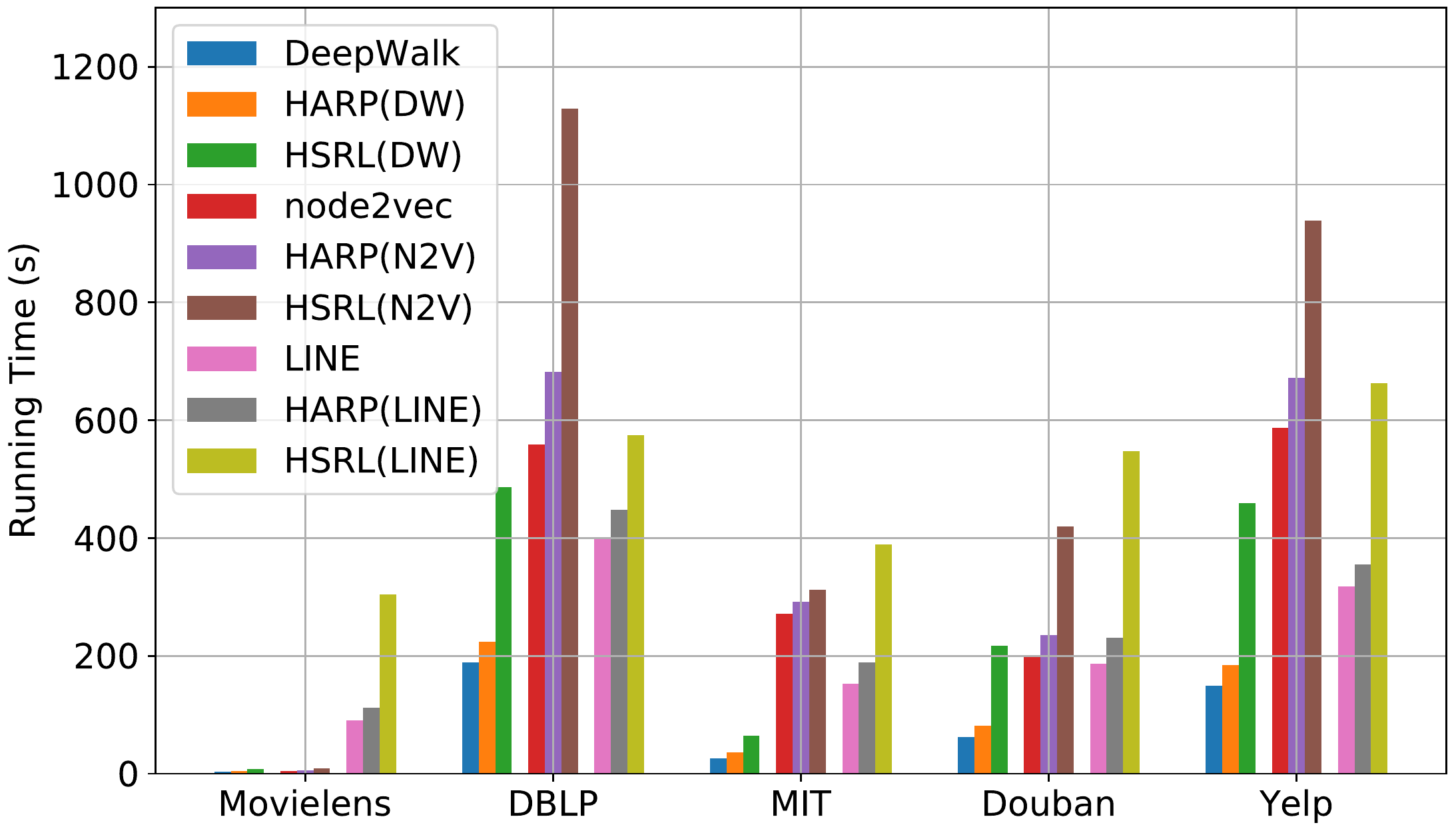}
	\caption{Running time}
	\label{figure7}
\end{figure}

\section{Conclusions}
Most conventional NRL methods aim to preserve the local topological information of a network but overlook their global topology. Recently, HARP was proposed to preserve both local and global topological information. However, it could easily get stuck at a bad local minimum due to its poor network compressing schemes. In this paper, we propose a new NRL framework, HSRL, to tackle these issues. Specifically, HSRL employs a community-awareness network compressing scheme to obtain a series of smaller networks based on an input network and conducts a NRL method to learn node embeddings for each compressed network. Finally, the node embeddings of the original network can be obtained by concatenating all node embeddings of compressed networks. Empirical studies on link prediction on various real-world networks demonstrate HSRL significantly outperforms the state-of-the-art algorithms.

Our future work includes combining HSRL with deep learning-based methods, such as DNGR\cite{cao2016deep}, SDNE\cite{wang2016structural}, and GCN\cite{kipf2016semi}. It is also very interesting to extend HSRL to learn node embeddings of more complex networks which may be more common in real-world applications, e.g., heterogeneous networks, attributed networks, and dynamic networks. 

\bibliographystyle{IEEEtran} 
\bibliography{reference} 
\end{document}